\documentclass[11pt]{article}

\def\full{1}
\def\ShowAuthNotes{0}

\title{\textsc{\textbf{
      Title
    }}
}

\author{
  Author
}
\date{\today}

\usepackage{color,graphicx}

\usepackage{amsmath,amstext}

\usepackage{amsfonts,amssymb}

\usepackage{amsthm}

\usepackage{textcomp}

\newtheorem{itheorem}{Informal Theorem}
\newtheorem{theorem}{Theorem}[section]

\newtheorem{claim}[theorem]{Claim}

\newtheorem{lemma}[theorem]{Lemma}
\newtheorem{corollary}[theorem]{Corollary}

\newtheorem{fact}[theorem]{Fact}

\theoremstyle{definition}
\newtheorem{definition}[theorem]{Definition}

\newtheorem{remark}[theorem]{Remark}

\usepackage{times}
\usepackage[varg]{txfonts} 
\renewcommand{\mathbb}{\varmathbb}

\usepackage{microtype}

\usepackage{url}

\usepackage{fullpage}

\usepackage{bm}

\usepackage[colorlinks,citecolor=blue,linkcolor=blue]{hyperref}

\usepackage{xspace}

\usepackage{dsfont}

\usepackage{nicefrac}

\newcommand{\nfrac}{\nicefrac}

\newcommand{\cA}{\mathcal A}

\newcommand{\cC}{\mathcal C}

\newcommand{\cF}{\mathcal F}
\newcommand{\cG}{\mathcal G}
\newcommand{\cH}{\mathcal H}

\newcommand{\cM}{\mathcal M}

\newcommand{\cP}{\mathcal P}
\newcommand{\cQ}{\mathcal Q}

\newcommand{\mper}{\,.}
\newcommand{\mcom}{\,,}

\definecolor{DSgray}{cmyk}{0,0,0,0.7}

\definecolor{DSred}{cmyk}{0,0.7,0,0.7}

\ifnum\ShowAuthNotes=1
\newcommand{\Authornote}[2]{\noindent{\small\textcolor{DSgray}{\sf{
\textcolor{red}{[#1: #2]\marginpar{\textcolor{red}{\fbox{\Large !}}}}}}}}
\newcommand{\Authormarginnote}[2]{\marginpar{\parbox{2.2cm}{\raggedright\tiny \textcolor{red}{#1: #2}}}}
\else
\newcommand{\Authornote}[2]{}
\newcommand{\Authormarginnote}[2]{}
\fi

\newcommand{\Dnote}{\Authornote{David}}

\newcommand{\Mnote}{\Authornote{Moritz}}

\newcommand{\Tnote}{\Authornote{Thomas}}

\newcommand{\Bnote}{\Authornote{Boaz}}

\newcommand{\Brac}[1]{\left[#1 \right]}

\newcommand{\Set}[1]{\left\{#1\right\}}

\newcommand{\norm}[1]{\lVert#1\rVert}

\newcommand{\from}{\colon}

\newcommand{\poly}{\mathrm{poly}}

\newcommand{\defeq}{\stackrel{\mathrm{def}}=}

\renewcommand{\leq}{\leqslant}
\renewcommand{\le}{\leqslant}
\renewcommand{\geq}{\geqslant}
\renewcommand{\ge}{\geqslant}

\newcommand{\vbig}{\vphantom{\bigoplus}}

\newcommand{\normt}[1]{\norm{#1}_{\scriptstyle 2}}
\newcommand{\varnormt}[1]{\norm{#1}_{\scriptscriptstyle 2}}

\newcommand{\N}{\mathbb N}
\newcommand{\R}{\mathbb R}

\newcommand{\sdp}{\mathrm{sdp}}

\newcommand{\opt}{\mathrm{opt}}

\newcommand{\Esymb}{\mathbb{E}}
\newcommand{\Psymb}{\mathbb{P}}

\DeclareMathOperator*{\E}{\Esymb}

\DeclareMathOperator*{\ProbOp}{\Psymb r}

\renewcommand{\Pr}{\ProbOp}

\newcommand{\Ex}[1]{\E\Brac{#1}}

\renewcommand{\epsilon}{\varepsilon}

\newcommand{\e}{\epsilon}
\newcommand{\eps}{\epsilon}

\newcommand{\sse}{\subseteq}

\renewcommand{\normt}{\varnormt}

\newcommand{\seteq}{\mathrel{\mathop:}=}
\newcommand{\la}{\leftarrow}

\usepackage{boxedminipage}

\newcommand{\using}[1]{\text{(using #1)}}

\newcommand{\maxcut}{\textsc{Max Cut}\xspace}
\newcommand{\gwsdp}{\textsf{GW SDP}\xspace}
\newcommand{\basicsdp}{\ensuremath{\mathsf{Basic SDP}}\xspace}
\newcommand{\basiclp}{\ensuremath{\mathsf{Basic LP}}\xspace}

\newcommand{\uniquegames}{\textsc{Unique Games}\xspace}

\newcommand{\remove}[1]{}

\newcommand{\normtv}[1]{\norm{#1}_{\scriptscriptstyle\mathrm{TV}}}

\newcommand\gep\succeq
\newcommand\lep\preceq
\newcommand\sa{\mathrm{lp}}

\newcommand\sdpGWT{\mathrm{sdp}\ensuremath{_3}}
\newcommand\sdpGW{\mathrm{sdp}}

\newcommand{\var}{\mathrm{Var}}
\newcommand{\Inf}{\mathrm{Inf}}

\newcommand{\torestate}[3]{%
\expandafter \def \csname BBRESTATE #2 \endcsname{#3}
\theoremstyle{plain}
\newtheorem{BBRESTATETHMNUM#2}[theorem]{#1}
\begin{BBRESTATETHMNUM#2}\label{#2}\csname BBRESTATE #2 \endcsname   \end{BBRESTATETHMNUM#2}
\newtheorem*{BBRESTATETHMNONNUM#2}{{#1}~\ref{#2}}
}

\newcommand{\restate}[1]{\begin{BBRESTATETHMNONNUM#1}[Restated] \csname BBRESTATE #1 \endcsname
\end{BBRESTATETHMNONNUM#1}}

\makeatletter

\ifnum\ShowAuthNotes=1
\newcommand\@BBnote[1]{{#1}}
\else
\newcommand\@BBnote[1]{}
\fi

\ifnum\full=1
\newcommand{\fullproceed}[2]{#1}
\newcommand{\onlyfull}[1]{\@BBnote{$<<<$ \bf Only in full:}#1\@BBnote{$>>>$}}
\newcommand{\onlyproceed}[1]{}

\newcommand{\defer}[1]{\@BBnote{$<<<$ \bf Defer (#1)}}
\newcommand{\deferred}[1]{\@BBnote{$<<<$ \bf Place here (#1) $>>>$}}

\else

\newcommand{\fullproceed}[2]{#2}
\newcommand{\onlyfull}[1]{}
\newcommand{\onlyproceed}[1]{\@BBnote{$<<<$ \bf Only in proceed:}#1\@BBnote{$>>>$}}

\newcommand{\defer}[1]{\@BBnote{$<<<$ \bf Text below deferred (#1) $>>>$}\expandafter\def\csname @BBDEFER#1 \endcsname}

\newcommand{\deferred}[1]{\@BBnote{$<<<$ \bf Deferred text (#1)}\expandafter\relax\csname @BBDEFER#1 \endcsname}

\fi

 \makeatother

\begin{document}

\newcommand\infn{\ensuremath{B_{\infty}^n}}

\sloppy

\title{\bf Subsampling Mathematical Relaxations \\ and  Average-case Complexity}

\author{Boaz Barak\thanks{Department of Computer Science and Center for Computational Intractability,
Princeton University \texttt{boaz@cs.princeton.edu}. Supported by NSF grants CNS-0627526, CCF-0426582 and CCF-0832797,
and the Packard and Sloan fellowships.}
\and Moritz Hardt\thanks{Department of Computer Science and Center for
Computational Intractability, Princeton University, {\tt mhardt@cs.princeton.edu}. Supported by NSF grants CCF-0426582
and CCF-0832797.}
\and Thomas Holenstein\thanks{Department of Computer Science, ETH Zurich,
\texttt{thomas.holenstein@inf.ethz.ch}. Work done while at Princeton University and supported by NSF grants CCF-0426582
and CCF-0832797.}
\and David Steurer\thanks{Department of Computer Science and Center for Computational Intractability,
Princeton University \texttt{dsteurer@cs.princeton.edu}. Supported by NSF grants 0830673, 0832797, 0528414.}}

\maketitle

\thispagestyle{empty}

\begin{abstract}
We initiate a study of when the value of mathematical relaxations such as
linear and semi-definite programs for constraint satisfaction problems (CSPs)
is approximately preserved when restricting the instance to a sub-instance
induced by a small random subsample of the variables.

Let  $\mathcal{C}$ be a family of CSPs such as 3SAT, Max-Cut, etc.., and let
$\Pi$ be a mathematical program that is a \emph{relaxation} for $\mathcal{C}$,
in the sense that for every instance $\cP\in\cC$, $\Pi(\cP)$ is a number in
$[0,1]$ upper bounding the maximum fraction of satisfiable constraints of
$\cP$. Loosely speaking, we say that \emph{subsampling holds} for $\cC$ and
$\Pi$ if for every sufficiently dense instance $\mathcal{P} \in \mathcal{C}$
and every $\e>0$, if we let $\mathcal{P}'$ be the instance obtained by
restricting $\mathcal{P}$ to a sufficiently large constant number of
variables, then $\Pi(\cP') \in (1\pm \e)\Pi(\cP)$. We say that \emph{weak
subsampling} holds if the above guarantee is replaced with $\Pi(\cP') =
1-\Theta(\gamma)$ whenever $\Pi(\cP)= 1-\gamma$, where  $\Theta$ hides only
absolute constants. We obtain both positive and negative results, showing
that:

\begin{enumerate}

\item Subsampling holds for the \basiclp and \basicsdp programs. \basicsdp is
a variant of the semi-definite program considered by Raghavendra (2008), who
showed it gives an optimal approximation factor for every
constraint-satisfaction problem under the unique games conjecture. \basiclp is
the linear programming analog of \basicsdp.

\item For tighter versions of \basicsdp obtained by adding additional
constraints from the Lasserre hierarchy, \emph{weak} subsampling holds for
CSPs of unique games type.

\item There are non-unique CSPs for which even weak subsampling fails for the
above tighter semi-definite programs.  Also there are unique CSPs for which
(even weak) subsampling fails for the Sherali-Adams \emph{linear programming}
hierarchy.

\end{enumerate}

As a corollary of our weak subsampling for strong semi-definite programs, we
obtain a polynomial-time algorithm to certify that \emph{random geometric
graphs} (of the type considered by Feige and Schechtman, 2002) of max-cut
value $1-\gamma$ have a cut value at most $1-\gamma/10$. More generally, our
results give an approach to obtaining average-case algorithms for CSPs using
semi-definite programming hierarchies.
\end{abstract}

\vfill
\pagebreak

\tableofcontents
\thispagestyle{empty}
 \pagebreak

 \setcounter{page}{1}

\section{Introduction}

In this paper we consider the following seemingly unrelated questions:
\begin{enumerate}

\item Is the \maxcut  problem hard on random geometric graphs of the type considered by Feige and
    Schechtman~\cite{FeigeSc02}?

\item Is the value of a mathematical relaxation for a constraint-satisfaction problem (CSP) preserved when one
    passes from an instance $P$ to a  random induced sub-formula of $P$?

\end{enumerate}

It turns out that (in a sense made precise below) the answer to the first question is ``no'' and in fact this is
intimately related to the second question. The answer to the second question is much more subtle, and, in contrast to
the case of the objective value\footnote{A $k$-CSP is a collection $\mathcal{P}$ of functions mapping $n$ variables
from some finite alphabet to $\{0,1\}$, such that every $P\in\mathcal{P}$ depends on at most $k$ variables. We define
the \emph{objective value} of a CSP to be the maximum of $\tfrac{1}{|\mathcal{P}|}\sum_{P\in\mathcal{P}}P(x)$, taken
over all possible assignments $x$ to the variables.} of the CSP, the answer strongly depends on the type of relaxation
and CSP.

\subsection{\maxcut  on the sphere}

\maxcut --- the problem of finding a cut maximizing the number of cut edges--- is a widely studied optimization
problem, important both in its own right, and as a testbed for techniques in algorithms and hardness of approximation.
The best approximation algorithm for \maxcut  known today is the semi-definite program \gwsdp of Goemans and
Williamson~\cite{GoemansWi95}, which is optimal in the worst-case under the unique games
conjecture~\cite{KhotKiMoOd07,MosselOdOl05}. \gwsdp is a special case of the \basicsdp algorithm for CSPs considered by
Raghavendra~\cite{Raghavendra08}, who showed that the latter algorithm always has an optimal approximation factor in
the worst-case under the unique games conjecture.

In particular \gwsdp gives a value of at most $1-\Omega(\e^2)$ when given as input a graph whose maximum cut cuts
$1-\e$ fraction of the edges.\footnote{As a relaxation for a maximization problem, the value of \gwsdp is always at
least as large as the integral objective value. Hence the fact that the relaxation outputs some value $v$ for an
instance $G$ is a \emph{certification} that the maximum cut of $G$ is at most $v$.} In this work we study the
\emph{average-case complexity} of \maxcut --- namely whether one can do better on natural distributions over  the
instances. Since random graphs are expanders and so obviously have a maximum cut value close to $1/2$ (and moreover
this fact can be efficiently certified using the second eigenvalue), one needs to consider other distributions over the
inputs. We consider \emph{random geometric graphs}, that in light of known results, arguably constitute the most
natural distribution of \maxcut instances that is not obviously easy.

\paragraph{Random geometric graphs.} A random geometric graph is obtained by taking the vertices as random unit vectors in $\R^d$,
and connecting two vertices $u,v \in \R^d$ based on their distance $\normt{u-v}$. We consider the distribution
$\cG_{n,d,\gamma}$, where the vertices are $n$ random unit vectors in $\R^d$, and we connect two vectors if
$\normt{u-v} \geq 2\sqrt{1-\gamma}$. By construction, \gwsdp will have value $1-\gamma$ on these graphs, but, as shown
by~\cite{FeigeSc02}, as long as $n$ is not too small these graphs will have with high probability a maximum cut value
of $1-c\sqrt{\gamma}$ for some absolute constant $c$. Moreover, as we observe here, for a suitable choice of $n$, these
graphs will also be hard instances for the Sherali-Adams~\cite{SheraliAd90} \emph{linear} programming hierarchies;
these are generally incomparable with \gwsdp and have been shown to solve \maxcut on \emph{dense}
graphs~\cite{VegaKe07}. Nevertheless, we show here that these graphs can be certified to have small max cut in
polynomial time. (A certification algorithm that the max-cut of a random graph from a distribution is at most $v$ is an
algorithm whose output always upper bounds the max-cut, and with high probability the output is at most $v$.)

\begin{itheorem}[\maxcut on random geometric graphs, see Theorem~\ref{thm:sdp3}]
\label{thm:maxcut} There is a polynomial-time algorithm that certifies that a random graph $G$ from $\cG_{n,d,\gamma}$
satisfies $\maxcut(G) \leq 1 - \Omega(\sqrt{\gamma})$, for every $\gamma \in (0,1)$, $d\in\N$ and $n \geq
C(\gamma)/\mu(\gamma,d)$, where $\text{\textsc{Max Cut}}(G)$ denotes the fraction of edges cut by the maximum cut in
$G$, $C(\gamma)$ is some constant depending only on $\gamma$, and $\mu(\gamma,d)$ denotes the normalized measure in the
unit sphere of the ball of radius $\sqrt{2\gamma}$ around some unit vector.
\end{itheorem}

By a simple calculation one can show that the probability that two random unit vectors $u,v$ in $\R^d$ will satisfy
$\normt{u-v} \geq 2\sqrt{1-\gamma}$ is exactly $\mu(\gamma,d)$, implying that if $n \ll 1/\mu(\gamma,d)$ the graph
$\cG_{n,d,\gamma}$ will have average degree $\ll 1$ (and hence has a trivial large max cut). Thus the value of $n$ that
Theorem~\ref{thm:maxcut} applies to is at most a constant factor larger than the minimum possible. The algorithm $A$ of
Theorem~\ref{thm:maxcut} is simply a tightening of the relaxation $\gwsdp$ obtained by adding the so-called ``triangle
inequalities'' to that program.

\subsection{Subsampling mathematical relaxations}

The other question we consider is whether the value of mathematical relaxations such as linear and semi-definite
programming is preserved under subsampling. That is, given a CSP instance $\phi$ on $n$ variables, we consider the
instance $\phi'$ obtained by choosing at random $S \subseteq [n]$ of some specified size, and keeping only the
constraints involving only variables in $S$. We ask in what cases the value of the relaxation of $\phi'$ is close to
the value of $\phi$.

This question is a variant of property testing~\cite{Ron00,Rubinfeld06} that we believe is interesting in its own
right. It also has algorithmic applications. Subsampling gives a fast way to ``sketch'' a CSP in a way that preserves
the the objective value but using a much smaller instance size. But since we generally cannot compute this objective
value in the worst case, we'd want to make sure that if $\phi$ was an ``easy instance'' for our algorithm, then $\phi'$
will be such an instance as well. A subsampling theorem for mathematical relaxations guarantees this property.

Subsampling for the objective value of constraint satisfaction problem (namely the fraction of satisfied constraints)
was studied before by Goldreich, Goldwasser and Ron~\cite{GoldreichGoRo98} who gave a subsampling theorem for \maxcut ,
and by Alon, de la Vega, Kannan and Karpinski~\cite{AlonVeKaKa03} who gave a subsampling theorem for general CSPs. But,
to our knowledge, subsampling for mathematical relaxations was not studied
before.\Mnote{it has been used in a couple of papers that the value of some LP
is preserved under subsampling, e.g., in Alon et al} As we show, unlike the case of the
objective value, subsampling sometimes fails for the value of relaxations, and this depends on the particular
relaxation and CSP.

Another, more minor difference between prior works and ours is that while prior works focused on the \emph{dense} case,
considering $k$-CSPs with  $\Omega(n^k)$ constraints, we consider general, possibly non dense, CSPs, and wish to
optimize the trade-off between the sample size and density. We say that a $2$-CSP
is $\Delta$-dense if every variable
appears in at least $\Delta$ constraints, and use a suitable generalization of this notion to $k$-CSPs (see
Section~\ref{sec:subsample-kcsp}). We show a subsampling theorem for the objective value of  $\Delta$-dense CSPs with
the optimal sample of size $O(n/\Delta)$. Namely, we show that the value of the induced instanced is equal to the value
of the original instance up to $1 \pm \e$ multiplicative factor, where $O$ notation in the sample size hides polynomial
factors in $1/\e$.  The only prior work to consider this trade-off was by Feige and Schechtman~\cite{FeigeSc02}, who
gave such a result for \maxcut with $O(n\log n/\Delta)$ sample size.

Our results for subsampling mathematical relaxations of CSPs are the following (see Section~\ref{sec:subsample-sdp}
and~\ref{sec:subsample-ug} for formal statements).  In all cases we consider a $\Delta$-dense CSP $\cP$ and a
subformula of $\cP'$ induced on a random subset of $\poly(1/\e)(n/\Delta)$ variables, and we let $\Pi(\cP)$ be the
value of the relaxation $\Pi$ on $\cP$.\footnote{For the positive results, our sample size is as small as possible; the
negative results hold also for much larger sample size and in particular show that one cannot get a constant size
subset even if $\Delta=\Omega(n)$, see Section~\ref{sec:negative}.}

We start by showing that subsampling holds for \basicsdp and \basiclp, where $\basicsdp$ is the semi-definite program
considered by Raghavendra~\cite{Raghavendra08} and \basiclp is its linear programming analog.\footnote{We actually use
the ``smoothed'' version of Raghavendra's SDP considered in
\cite{RaghavendraSt09a,Steurer10}--- see
Section~\ref{sec:subsample-sdp}. The two programs are closely related,
and~\cite{Raghavendra08}'s result holds for the
smoothed version as well.}

\begin{itheorem}[Subsampling for \basicsdp and \basiclp, see Section~\ref{sec:subsample-sdp}]
\label{item:subsample-basic} In the notation above, for any CSP $\cP$ and for $\Pi$ that is either $\basicsdp$ or
$\basiclp$,
\[
 \Pi(\cP) - \e \leq \Pi(\cP') \leq \Pi(\cP) + \e
\]
\end{itheorem}

We then show that for stronger SDPs, we still have weak subsampling if the CSP is a unique game.

\begin{itheorem}[Subsampling for unique games, see Theorem~\ref{thm:subsample-ug}] \label{ithm:subsample-ug}
In the notation above, if $\cP$ is a unique game,
then for every $k\in\N$, letting  $\gamma = 1- \basicsdp_k(\cP)$,
 \[
1 - \gamma - \e \leq \basicsdp_k(\cP') \leq 1-\gamma/9 + \e
\] where $\basicsdp_k$ denotes \basicsdp augmented with $k$ rounds of the
Lasserre hierarchy.
\end{itheorem}

Theorem~\ref{ithm:subsample-ug} is the main technical contribution of this paper, and also the one used to obtain our
algorithm for \maxcut on random geometric graphs. We also have negative results that complement our positive results
and show that, in contrast to the case of the objective value, subsampling sometimes fails for mathematical
relaxations.

\begin{itheorem}[Negative results for subsampling, see Theorems~\ref{thm:no-subsample-sdp} and~\ref{thm:sa}]
There is a (non unique) CSP $\cP$ and absolute constant $\delta>0$ for which $\basicsdp_{O(1)}(\cP)\leq 1-\delta$ but
with high probability $\basicsdp_{\sqrt{n}}(\cP') \geq 1-o(1)$.  There is a unique CSP $\cP$ and absolute constant
$\delta>0$ for which $\basiclp_3(\cP) \leq 1-\delta$ but with high probability $\basiclp_{\omega(1)}(\cP) \geq 1-o(1)$,
where $\basiclp_k$ denotes \basiclp augmented with $k$ rounds of the Sherali-Adams hierarchy, and $o(1)$ (resp.
$\omega(1)$) denotes a function that tends to $0$ (resp. $\infty$) with $n$.
\end{itheorem}

See Figure~\ref{fig:subsampling} for an overview of our positive and negative results on subsampling mathematical
relaxations. As one can see, we cover most of the cases, with the most interesting (in our opinion) open question is
whether strong semidefinite programs for unique CSPs such as \maxcut actually have \emph{strong} subsampling, in the
sense that the value of the program on a subsample approximates the value on the original instance with arbitrary
accuracy. We suspect that the answer is ``no'', though have no proof of that.

\begin{figure}
\begin{center}
\includegraphics{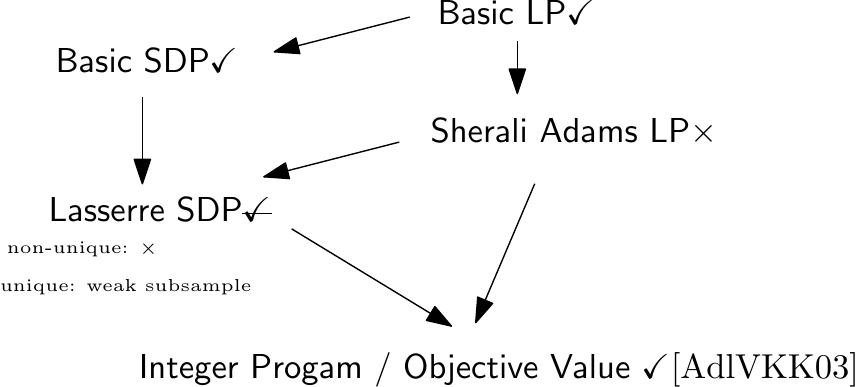}
\caption{Overview of subsampling results. {\footnotesize \checkmark denotes a subsampling theorem, while $\times$ denotes
that subsampling fails. Arrows point from weaker to tighter relaxations. } }
\label{fig:subsampling}
\end{center}
\end{figure}

\subsection{Subsampling SDPs and average-case complexity.}

As mentioned above, we use Theorem~\ref{ithm:subsample-ug} (weak subsampling theorem for strong SDPs on unique CSPs) to
show Theorem~\ref{thm:maxcut}--- the \maxcut algorithm for random geometric graphs. Theorem~\ref{thm:maxcut} is
obtained from our subsampling theorem as follows: we first show that $\basicsdp_3(G_{d,\gamma}) \leq 1-
\Omega(\sqrt{\gamma})$ where $\basicsdp_3$ denotes the $\basicsdp$ program augmented with the triangle inequalities,
and $G_{d,\gamma}$ is the graph on the continuous $d$-dimensional sphere where we connect two unit vectors $u,v$ if
$\normt{u-v} \geq 2\sqrt{1-\gamma}$. (Equivalently, one can think of $G_{d,\gamma}$ as a random graph from
$\cG_{n,d,\gamma}$ where $d,\gamma$ are fixed and $n$ tends to infinity.) We show this by observing that the edges of
$G$ can be partitioned into essentially disjoint union of \emph{odd cycles} of length $O(1/\sqrt{\gamma})$, and
noting that triangle inequalities can capture the fact that  one cannot cut all the edges of an
odd cycle. Since random geometric graphs are simply subsamples of $G_{d,\gamma}$, our subsampling theorem implies that
$\basicsdp_3$ will have value in $1-\Theta(\sqrt{\gamma})$ for these graphs.

This algorithm is an instance of a general recipe for using our subsampling theorem for average-case algorithms. Many
natural distributions can be thought of as random subsamples of some instance (or family of instances)~$\phi$ (e.g.,
random graphs are subsamples of a random dense graph, random 3SAT are subsamples of a random dense formula). In such
cases, if one can give a relaxation that gives a tight value on $\phi$ (perhaps by exploiting its density) and the
relaxation admits subsampling, then it follows that the relaxation succeeds on the distribution of subsamples as well.
In our case, even though sufficiently many rounds on Sherali-Adams hierarchy give a tight value on $G_{d,\gamma}$
(since, considering $\gamma$ as constant, it is dense), we cannot use those directly as they do not admit subsampling.
Similarly, even though $\basicsdp$ admits subsampling, it does not yield a tight value on dense 3SAT formulas, which is
the reason our results do not refute Feige's hypothesis~\cite{Feige02} on the hardness of certifying that random 3SAT
formulas are unsatisfiable.

We note that subsampling theorems have been used before for approximation
algorithms for CSPs, but in a different way.  Prior works used subsampling of
the objective value to show worst-case approximation algorithms for dense
graphs, by showing that one can first subsample to constant size and then
solve the problem using brute force on the
sample~\cite{AlonVeKaKa03}
(or use that argument to show that
linear programming hierarchies will succeed on the original
instance~\cite{VegaKe07}). In contrast we use subsampling of the relaxation
value to give average-case algorithms on some specific distributions of
(possibly sparse) graphs. Our result is also one of the few examples where
higher order SDPs can succeed in an algorithmic task in which \basicsdp fails.
As mentioned above, if the unique games conjecture is true, then \basicsdp is
an optimal worst-case approximation algorithm for CSPs, though of course it
can be worse than other efficient algorithms on some (distributions of)
inputs.

\subsection{Related work}

As mentioned above, there has been many works on estimating graph parameters from random small induced subgraphs of
\emph{dense} graphs. Goldreich, Goldwasser and Ron \cite{GoldreichGoRo98} show that the the Max-Cut value of a dense
graph (degree $\Omega(n)$) is preserved by subsampling. (In this and other results, the constants depend on the quality
of estimation.)  Feige and Schechtman~\cite{FeigeSc02} showed that the result holds generally for $\Delta$-dense graphs
so long as the degree $\Delta \geq \Omega(\log n)$ and the subgraph is of size at least $\Omega(n\log n/\Delta)$. (As a
corollary of our results, we slightly strengthen \cite{FeigeSc02}'s bounds to hold for any $\Delta>\Omega(1)$ and
subgraph size larger than $\Omega(n/\Delta)$.) Alon et al~\cite{AlonVeKaKa03} generalize \cite{GoldreichGoRo98} for
$k$-CSP's and improve their quantitative estimates. See also \cite{RudelsonVe07} for further
quantitative improvements in the case of $k=2$.  

There has also been much work on matrix and graph sparsification by means other than \emph{uniform} sampling, see for
instance~\cite{SpielmanTe04,AroraHaKa06,AchlioptasMc07, SpielmanSr08,BatsonSpSr09}. Indeed, \emph{spectral sparsifiers}
are stronger than the notion we consider, in the sense that passing to a spectral sparsifier will preserve the SDP
value for, say, \maxcut . Algorithmically though, if one only wants to
preserve the SDP value, there are some advantages to subsampling, as it
reduces not just the number of edges but also the number of vertices, hence
potentially yielding \emph{sublinear} algorithms, and can also be carried out
very efficiently by just random sampling,
reducing to a subgraph of constant degree. In contrast constant degree spectral sparsification~\cite{BatsonSpSr09}
cannot be achieved by sampling vertices (or even edges for that matter) uniformly at random, even for regular graphs.

\section{Overview of proofs} \label{sec:overview}

In this section we give a high level overview of our proofs, focusing on our main result---
Theorem~\ref{ithm:subsample-ug} showing a weak subsampling theorem for strong semidefinite programs for unique games. A
$k$CSP $\mathcal{P}$ on an alphabet $[q]$ is a collection of local functions (called ``constraints'') from $[q]^n\to
[0,1]$, where for $x\in [q]^n$ we denote $\cP(x) = \tfrac{1}{|\cP|}\sum_{P\in\cP}P(x)$. If $U$ is a set of variables,
then $\cP[U]$ denotes the restriction of $\cP$ to those constraints that depend on variables in $U$. We'll let $\cP'$
denote $\cP[U]$ where $U$ is a random subset of size set to an appropriate parameter (that we ignore in this overview).

\subsection{Subsampling for $k$-CSPs and \basicsdp} \label{sec:overview:kcsps}

Alon, de la Vega, Kannan and Karpinski~\cite{AlonVeKaKa03} proved a subsampling theorem for $k$-CSP. As a first step,
we extend their results to hold with a better dependency between the sample size and density, and to hold for
constraints that can output a real number, say in $[0,1]$, rather than just a Boolean value. The latter extension is
trivial, but the former (which we need for our \maxcut application) requires some work, adapting and refining
techniques of \cite{GoldreichGoRo98,FeigeSc02,VegaKe07}.  Our subsampling theorem for (generalized) $k$-CSPs is stated
in Section~\ref{sec:subsample-kcsp} and proven in Section~\ref{sec:proof-kcsp}.

\paragraph{Subsampling for \basicsdp.} \basicsdp is the semi-definite program for $k$-CSPs
considered by Raghavendra~\cite{Raghavendra08}, who showed that it gives an optimal  approximation ratio in the
worst-case under the unique games conjecture. For a given $k$CSP $\cP$ over alphabet $[q]$, this program assigns a
vector $v_{i,a}$ for every variable $x_i$ and alphabet symbol $a\in [q]$ of
$\cP$. It also assigns $q^k$ numbers $\mu_{P,1^k},\ldots,\mu_{P,q^k}$ for
every constraint $P$ of $\cP$. It makes the following consistency requirement
on
$\{ v_{i,a} \}$ and $\{ \mu_{P,x} \}$--- the inner product of $v_{i,a}$ and $v_{j,b}$ should match the probability of
the event ``$x_i=a$ AND $x_j=b$'' in any local distribution $\mu_P$ involving both variables $x_i$ and $x_j$ (this can
be captured by linear and semi-definite conditions).  The value of the CSP is simply the expectation of $P(x)$ over a
random constraint $P$ and a random partial assignment $x$ chosen from $\mu_P$. (To avoid the potential issue of the SDP
being extremely sensitive to few of the constraints, we
follow~\cite{RaghavendraSt09a,Steurer10} in allowing a bit of
slackness in the consistency constraints on $\mu_P$.)

Our subsampling theorem for \basicsdp, proven in Section~\ref{sec:subsample-sdp}, follows from the general subsampling
for $k$-CSPs. The idea is to combine two observations: (1) because the assignment to the vectors $\{ v_{i,a} \}$
determines the best choice for the local distribution, it is possible to write $\basicsdp$ as a program that has no
constraints and needs to maximize a sum of local functions over these vectors, (2) one can use dimension reduction to
assume that the vectors have \emph{constant} dimension with little loss of
accuracy~\cite{RaghavendraSt09a}. Thus by
discretizing this constant dimensional space, we can think of $\basicsdp$ as itself a CSP over some constant sized
alphabet, and apply our $k$-CSP subsampling theorem to this CSP. A similar (even simpler) reasoning applies to the
linear programming variant $\basiclp$, and also to quadratic programs, in particular implying a variant of property
testing for positive semi-definiteness, see Section~\ref{sec:subsample-kcsp}.

\subsection{Weak subsampling for strong SDPs} \label{sec:overview:weak}

We now give a high level overview of the proof of Theorem~\ref{ithm:subsample-ug}. Because stronger SDPs such as those
from the Lasserre hierarchy actually involve constraints including several vectors, they cannot be expressed as a CSP
in the same way as \basicsdp. Indeed, we have negative results showing that subsampling can fail for these SDPs (see
Sections~\ref{sec:overview:neg} and~\ref{sec:negative}).

The result actually does not depend on the particular properties of the Lasserre hierarchy, and holds for a very
general class of strong SDPs. We start by formalizing this class. Any strong SDP can be thought of as the program
\basicsdp augmented with the constraint that the positive semi-definite matrix $X$ of the inner products of all these
vectors is in some convex set $\cM$. But one needs the set $\cM$ to satisfy some ``niceness conditions'' in order for
it to make sense to apply the program on a subsampled CSP. The niceness conditions we consider are rather mild, and
require that solutions remain valid under renaming and identifying of vertices (see
Section~\ref{sec:subsample-ug}). In particular they apply to any SDP obtained by adding a number of Lasserre rounds
to \basicsdp .

If $\Pi$ is any strong SDP, $\cP$ is a CSP, and $\cP'$ is a subsample of $\cP$, then it's not hard to show that with
high probability $\Pi(\cP') \geq \Pi(\cP) - \e$, since this only needs the
argument that the value of one solution
(the optimal one for $\cP$) will be approximately preserved. The challenging task is to show that $\Pi(\cP')$ is not
much larger than $\Pi(\cP)$, and because subsampling does not always hold for SDPs, we know that the proof for
subsampling of $k$-CSPs does not generalize to this case.

The crucial notion we use is of that of a \emph{proxy CSP}. Let $\cG$ and $\cH$ be two unique games on the same
alphabet and number of variables, we say that $\cH$ is \emph{proxy} for $\cG$ (with respect to the program $\Pi$), if
for \emph{every} assignment $X$ (even possibly outside $\cM$) to the vectors of $\Pi$, $1-\Pi(\cG)[X] \leq
1-\Pi(\cH)[X]/10$, where $\Pi(\cP)[X]$ denotes the value of the program $\Pi$ on the CSP $\cP$ with assignment $X$ to
the vectors of $\Pi$.\footnote{The actual domination condition we use will restrict the possible vector assignments
based on the norms of the vectors, but because we restrict the vectors to a product set, it does not make a difference
in our arguments.} That is, one can think of $\cH$ as pointwise dominating $\cG$ with respect to the program $\Pi$. We
then show that this domination condition is somewhat preserved under subsampling, at least for the optimal solutions.
That is, we show that with high probability $1- \Pi(\cG')  \leq 1 -\Pi(\cH')/10 + \e$, where $\cG'$ and $\cH'$ are the
subsampled versions of $\cG$ and $\cH'$. The idea here is to use our subsampling theorem for SDP looking at the SDP
$\max_X \Pi(\cH)[X]/10 - \Pi(\cG)$. This is a basic SDP since it places no constraints on $X$, and so since we know its
optimum is at most $0$, this should be approximately preserved under subsampling.

The above discussion shows that to prove Theorem~\ref{ithm:subsample-ug} it suffices to find some unique game $\cH$
such that (*) $\cH$ is a proxy for $\cG$ and (**) with high probability $1-\Pi(\cH') \leq 1 - \Pi(\cG) + \e$. This is
what we do. The proxy game $\cH$ is simply the game $\cG^3$ obtained by taking all length-$3$ paths in the constraint
graph of $\cG$ and composing the corresponding permutations. Condition (*) is not that hard to show. Intuitively, an
assignment that satisfies $1-\gamma$ fraction of the constraints of $\cG$ should satisfy about $1-3\gamma$ fraction of
the constraints of $\cH$, (since each one is just three constraints of $\cG$) and this reasoning carries over to SDP
assignments as well.

Condition (**) looks suspiciously close to what we're trying to prove in the first place (preservation of value under
subsampling), but note the asymmetry--- we need to show that a subsample of $\cH$ will have roughly the same value as
the original graph $\cG$. It turns out this will actually help us. What we need to show is a way to decode an
assignment for the SDP of the subsampled game $\cH'$ into an assignment of roughly the same value for the SDP of the
original game $\cG$. For simplicity, assume that the alphabet of the CSP is $\{0,1\}$ in which case the vector
assignment is just one vector per variable.\footnote{Although in the phrasing above it seems that one would need two
vectors per variable for alphabet of size $2$, it is known how to transform the SDP into an equivalent program needing
only one vector per variable in this case.} Suppose that $\cG$ has $n$ variables, each participating in $\Delta$
constraints, and we subsample to a set $S$ of size $n' = O(n/\Delta)$ variables.\footnote{Note that, ignoring constant
factors, $\cG$ has roughly $n\Delta$ constraints, $\cG'$ has $n/\Delta$ constraints, $\cH$ has $n\Delta^3$ constraints,
and $\cH'$ has $n\Delta$ constraints--- the latter fact is some indication why one may hope to decode an assignment to
$\cH'$ into an assignment to $\cG$.} We are given a vector assignment $\{ v'_{i'} \}_{i'\in S}$ for each of the $n'$
variables in the sample that gives value $\tau$ for $\Pi(\cH')$, and need to ``decode'' it into an assignment $\{ v_i
\}_{i\in [n]}$ that gives value roughly $\tau$ for $\Pi(\cG)$. We will use a randomized decoding, assigning for every
variable $i$ of $\cG$  the vector $v_{i'}$ where $i'$ is a random neighbor of $i$ in $\cG$ that is contained in the
sample $S$.\footnote{Our ``niceness conditions'' will ensure that if the inner-product matrix of the original
assignment was in $\cM$ then the same will hold for the decoded assignment. Also we will flip the vector if the
corresponding permutation on $\{0,1\}$ was $a \mapsto -a$, but in the discussion below as assume that all permutations
involved are the identity--- this simplifies notation and is immaterial to this argument.} Let $(i',i,j,j')$ be the
length-$3$ path corresponding to a random constraint of $\cH'$ that survived the subsampling. That is, $i',j' \in S$.
If the subsampled graph is (approximately) regular, we can choose $(i',i,j,j')$ in the following way: first let $(i,j)$
be variables corresponding to a random constraint in $\cG$, then take $i'$ to be a random neighbor of $i$ that is also
in $S$, and take $j'$ to be a random neighbor of $j$ that is also in $S$. We know that on average the vectors $v'_{i'}$
and $v'_{j'}$ contribute $\tau$ to the value of $\Pi(\cH')$. But then on expectation the contribution to $\Pi(\cG)$ of
the decoded vectors $v_i$ and $v_j$ is also $\tau$, since $v_i$ is exactly obtained by taking $v'_{i'}$ for a random
neighbor $i'\in S$ of $i$, and $v_j$ is obtained by taking $v'_{j'}$ for a random neighbor $j'\in S$ of $j$. This
concludes the proof. We remark that this reasoning is somewhat reminiscent of Dinur's analysis of her gap amplification
lemma for PCP's~\cite{Dinur07}.

\subsection{Negative results for subsampling} \label{sec:overview:neg}

We now briefly sketch why, unlike the case for $k$CSPs, subsampling sometimes
\emph{fails} for strong semidefinite and linear programs---- see
Section~\ref{sec:negative} for more details. The idea is simple: many
integrality gaps examples, for both LP hierarchy and SDP's, are actually
obtained from random instances. Examples include Schoenebeck's
result~\cite{Schoenebeck08} showing random 3SAT is an integrality gap example
for the Lasserre SDP hierarchy, and results showing that random graphs (and
more generally good expanders) are integrality gap examples for linear
programming hierarchies for \maxcut~\cite{VegaKe07,CharikarMaMa09}. Such
random instances can be thought of as subsampling of sufficiently dense
instance. But sufficiently strong SDP or LP  programs will succeed in
certifying that a dense instance has small value.
Thus these
integrality gaps give example of a CSP $\cP$  where $\Pi(\cP)$ is small, where
$\Pi$ is a sufficiently strong linear or semidefinite program, but $\Pi(\cP')$
is close to $1$ for a random induced sub-instance $\cP'$ of $\cP$. Note that
indeed for unique games random graphs are actually easy for semi-definite
programs~\cite{AroraKhKoStTuVi08}, explaining perhaps why subsampling for
unique games is possible for semi-definite programs but not for linear
programs.

\section{Preliminaries}

Let $G$ be a $\Delta$-regular graph with vertex set $V=[n]$ and edge
set $E$ (no parallel edges or self-loops).
We give weight $\nfrac2{\Delta n}$ to each edge of $G$ so that every
vertex of $G$ has (weighted) degree $2/n$ and $G$ has total edge
weight $1$.
We say a graph is \emph{normalized} if it has total edge weight $1$.
(We choose this normalization, because we will often think of a graph
as a probability distribution over unordered vertex pairs.)
For a graph $G$ as above and a vertex subset $U \subseteq V$, let
$G[U]$ denote the \emph{induced subgraph} on $U$.
To preserve our normalization, we scale the weights of the edges of
$G[U]$ such that the total edge weight in $G[U]$ remains~$1$.
We denote by $V_{\delta}$ a random subset of a $\delta$ fraction of
the vertices in $V$, and hence $G[V_{\delta}]$ denotes a random
induced subgraph of $G$ of size $\delta|V|$.
With our normalization, the typical weight of an edge in
$G[V_\delta]$ is $\nfrac 2{\delta^2\Delta n} $.
\ifnum\full=0
\vspace{-3mm}
\else
\fi

\paragraph{Max $k$-CSPs.}
A $k$-CSP instance ${\cal P}$ is a set of predicates (or pay-off functions)
of the form
$P\colon[q]^n\to\R,$ where every $P=P(x_{i_1},x_{i_2},\dots,x_{i_k})$ is a $k$-junta,
meaning it depends only on $k$ of the $n$ variables in $x.$ We'll think of
$\var(P)=(i_1,\dots,i_k)$ as an ordered set and denote the $r$-th variable by
$\var_r(P)=i_r.$ Without loss of generality we may assume that in each
predicate $P\in\cP,$ all $k$ variables are distinct.
The norm of a pay-off function is defined as $|P|\defeq\max_{x\in[q]^n}|P(x)|\mcom$ and we put
$|\cP|=\sum_{P\in\cP}|P|.$ \Bnote{Should this be expectation?}

We think of~$\cP$ itself as a mapping
$\cP\colon[q]^n\to\R$ defined as
${\cal P}(x) \defeq \frac1{|\cP|}\sum_{P\in\cP} P(x)\mper$
The optimum is denoted $\opt(\cP)=\max_{x\in[q]^n}\cP(x).$
We will typically assume that
$|P|\le1$ for all $P\in\cP$ in which case $\opt(\cP)\le 1.$
For a subset $U\subseteq[n],$ with $|U|=\delta n,$
we let $\cP_U$ denote the~$k$-CSP
$\cP_U = \{ \delta^{-k}P \colon P\in \cP, \var(P)\subseteq U \}\mper$
In this case, we define $\cP_U(x)=\frac1{|\cP|}\sum_{P\in\cP_U}P(x)\mper$

\ifnum\full=0
\vspace{-3mm}
\else
\fi
\paragraph{Unique Games.}
A unique game $\cG$ is given by a constraint graph $G=(V,E),$ an alphabet~$[R]$ and constraints $\pi_{v\la u}$ for each
edge $e=(u,v)\in E.$ An assignment $x\in[R]^n$ satisfies the edge $e$ if $\pi_{v\la u}(x_u)=x_v.$ It will be convenient
for us to define unique games as a \emph{minimization} problem in which the objective is to minimize the number of
unsatisfied constraints. Note that throughout the introduction \uniquegames
was a maximization problem, but these two views are equivalent. As a minimization problem unique game has the following SDP relaxation (which is closely
related to \basicsdp program mentioned before):
\begin{equation}\label{eq:unique-sdp}
\sdp(\cG)\defeq
\min
\E_{(u,v)\in E} \E_{a\in[R]} \|u_a - v_{\pi_{v\la u}(a)}\|^2
\end{equation} subject to the constraints that $\sum_{a\in[R]}\|u_a\|^2=1$ for
every $u\in V$ and
$\langle u_a,u_b\rangle = 0$ for all $u\in V$ and $a\ne b$.
An SDP solution is a positive semidefinite
$(V\times[R])\times(V\times[R])$ matrix written as
$(X_{ua,vb})_{u,v\in V,a,b\in[R]}$ so
that $X_{ua,vb}=\langle u_a,v_b\rangle.$
We will denote by $\cM_2$ the set of such matrices that satisfy the three
constraints above. We will write $\sdp(\cG)[X]$ to denote the value of
$\sdp(\cG)$ under the particular solution~$X.$ We denote by $\cG[U]$ the
unique game $\cG$ restricted to the constraint graph $G[U].$
%


\section{Subsampling theorem for Max-$k$CSPs}
\label{sec:subsample-kcsp}

We will now state our subsampling theorem for $k$-CSPs and, as direct application, obtain subsampling theorems for
basic semidefinite relaxations of $k$-CSPs. To state the theorem we need a notion of \emph{density} of a $k$-CSP. For
$2$-CSPs we will use the standard notion of density in a graph. Specifically, we will say a $2$-CSP is
\emph{$\Delta$-dense} if every vertex has $\Theta(\Delta)$ neighbors. For $k$-CSPs when $k>2$
a natural generalization is to demand that after assigning $k-1$ out of $k$ coordinates in each constraint, there are
still $\Theta(\Delta)$ constraints remaining. In this case we say that the $k$-CSP is $\Delta$-dense.
\begin{theorem}\label{thm:subsample-kcsp-1}
Let $\eps>0,\Delta>1.$
Let $\cP$ be a $\Delta$-dense $k$-CSP in $n$ variables over an alphabet of
size $q$ so that $|P|\le1$ for all $P\in\cP.$
Put $\delta\ge\eps^{-C}\log(q)/\Delta$ for some absolute constant $C.$
Suppose $U\subseteq[n]$ is chosen uniformly at
random so that $|U|=\delta n.$ Then,
\begin{equation}
\left| \E\opt(\cP_U)-\opt(\cP) \right| \le\epsilon \mper
\end{equation}
\end{theorem}
The formal density condition and the proof of this theorem are given in
Section~\ref{sec:proof-kcsp}. We instead proceed to discuss the applications
of this theorem.

\subsection{Subsampling basic semidefinite programs}
\label{sec:subsample-sdp}

The above subsampling theorem for $k$-CSPs can actually be used to give a general subsampling theorem for \emph{basic}
semidefinite programs. A semidefinite program is called \emph{basic} if it can be written as a $2$-CSP $\cQ$ in $n$
variables (over infinite alphabet) of the following form:
\begin{equation}\label{eq:sdp-form}
\opt(\cQ)
=\max\E_{i,j\in[n]}P_{ij}(\{v_{i,a}\}_{a\in[R]},\{v_{j,b}\}_{b\in[R]})
\end{equation}
where $\cQ=\{P_{ij}\}_{i,j\in[n]}$ so that each $P_{ij}$ is a continuous function satisfying a Lipschitz condition in
the inner products $\langle v_{i,a},v_{j,b}\rangle.$ Here the maximum is taken over a bundle of $R$ vectors
$\{v_{i,a}\}$ per variable $i\in[n].$ We further require that each constraint on the vectors involves only vectors from
the same bundle $\{v_{i,a}\}_{a\in[R]}$ (such as $\|v_{i,a}\|^2\le1$ or $\sum_{a\in[R]}\|v_{i,a}\|^2=1$). We also
assume that~$|P_{ij}|\le1.$

It is crucial here that the
maximization is over a \emph{product} space of $n$ coordinates. Each
coordinate corresponds to one vector bundle $\{v_{i,a}\}_{a\in[R]}.$ Still we
cannot yet apply our subsampling theorem, because each coordinate is maximized
over a continuous space, i.e., $(B_2^d)^R.$ However, using dimension reduction as
in~\cite{RaghavendraSt09a}, the dimension of the vectors can be assumed to be
$\poly(1/\eps)$ without changing the objective value by more than an
$\eps/2$. Once the dimension is small we can discretize the space by an
$\eps'$-net (for small enough~$\eps'$) changing the inner products again only
by~$\eps/2.$ Hence we have the following lemma.
\begin{lemma}\label{lem:dim-net}
Let $\cQ$ be a $\Delta$-dense $2$-CSP of the form~(\ref{eq:sdp-form}). Then
there is a $2$-CSP $\cQ'$ with alphabet size at most $2^{\poly(\nfrac1\eps)}$
such that $|\opt(\cQ)-\opt(\cQ')|\le\eps.$
\end{lemma}
This shows that we do have a strong subsampling theorem for any \emph{basic}
semidefinite program:
\begin{corollary} \label{cor:subsample-sdp}
Let $\cQ$ denote a basic semidefinite program.
Assume $\cQ$ is $\Delta$-dense and let $\eps>0.$
Then,
\begin{equation}
\left|
\E\opt(\cQ_U)-\opt(\cQ)
\right|\le \eps\mcom
\end{equation}
where $U\subseteq[n]$ is a randomly chosen set of
size $\eps^{-C}n/\Delta$ for sufficiently large $C>0.$
\end{corollary}
\begin{proof}
After applying Lemma~\ref{lem:dim-net}, we can use
Theorem~\ref{thm:subsample-kcsp-1} to conclude the claim. Note that the alphabet
size of $2^{\poly(1/\eps)}$ translates into a factor $\poly(1/\eps)$ in
sample size.
\end{proof}
We will next demonstrate that both \basicsdp for $k$-CSPs and
the \uniquegames SDP are in
fact basic relaxation of the above form and therefore have a strong
subsampling theorem.

For the \uniquegames SDP this is immediate after changing it from a minimization problem to a maximization problem. (We
can simply multiply the objective by $-1$.) Note that the SDP relaxation for unique games corresponds to a dense
$2$-CSP if this is the case for the constraint graph. We remark that the same is true for the difference of two dense
\uniquegames relaxations and this is the case that will be used in the proof of our main theorem later
(Section~\ref{sec:subsample-ug}).

More generally, the same can be done for the \basicsdp relaxation of any $k$-CSP. Raghavendra~\cite{Raghavendra08}
defined  \basicsdp for a $k$-CSP $\cP=\{P_1,\dots,P_m\}$ with $|P_t|\le1$ over the alphabet~$[R]$ as the program
\[
\max\E_{t\in[m]}\E_{x\sim\mu_t}P_t(x)
\]
subject to the constraint that $\Pr_{x\sim\mu_t}\{x_i=a,x_j=b\}=\langle v_{i,a},v_{j,b}\rangle$ for all $t\in[m],$
$i,j\in\var(P_t),$ and $a,b\in[R].$ The maximum is taken over all ensembles $\{v_{i,a}\}$ of unit vectors and $R^k$
tuples of variables $\mu_t$, each of which is required to be a probability distribution on $\var(P_t).$ Let
$\mathsf{violate}(t)$ denote the sum of $|\Pr_{x\sim\mu_t}\{x_i=a,x_j=b\}-\langle v_{i,a},v_{j,b}\rangle|$ over all
$i,j \in \var(P_t)$ and $a,b \in R$. While the constraints of~\cite{Raghavendra08} is that $\mathsf{violate}(t)=0$ for
all $t\in [m]$, we follow~\cite{RaghavendraSt09a,Steurer10} that replaced this with the constraint $\E_{t \in [m]}
\mathsf{violate}(t) \leq \e$ and showed the two programs are approximately equivalent up to $\poly(\e)$ perturbation of
the instance. As shown in~\cite{Steurer10}, because there are only a few ($R^2k^2$) constraints per pay-off
function $P_t$, we can introduce this penalty function into the objective function, adding the term $-\E_{t \in [m]}
\mathsf{violate}(t)/\e$ into the expression we maximize. Hence, for our purposes we may assume that \basicsdp has the
form in~(\ref{eq:sdp-form}) so that our subsampling theorem applies. We stress that this approach can only work since
there are a few constraints for each pay-off functions of~$\cP.$ The approach breaks down in the presence of
constraints that involve arbitrary combinations of variables, such as $\ell_2^2$ triangle inequalities. In this case it
is no longer possible to assign a meaningful penalty to each constraint.

In the case of \basiclp similar arguments apply. \basiclp is the same as
\basicsdp except that we don't require the probability distributions to be
realized as inner products of vectors. Two distributions $\mu_P$ and
$\mu_{P'}$ are however required to be consistent whenever they share a
variable. These constraints can be written in the objective function and this
results in a $2$-CSP to which our subsampling theorem applies.

\paragraph{Application to property testing positive definite matrices.}
Our subsampling theorem also applies to quadratic forms and this can be very
useful. We illustrate one application in the context of property testing.
Specifically, we will get a property testing algorithm for
the class of positive semidefinite matrices.
Let us say that a matrix $B$ is \emph{$\eps$-far} from positive
semidefinite definite if there exists a vector $x$ with $\|x\|_\infty\le1$
such that $-\eps\ge\langle x,Bx\rangle=\sum_{ij}b_{ij}x_ix_j.$
Recall that $B$ is positive semidefinite if and only if $\langle x,Bx\rangle\ge0$
for all $x.$
Notice we could have defined distance in terms of the operator
norm which is to say that there exists an $x$ with $\|x\|_2\le1$ such that
$\langle x,Bx\rangle\le-\eps.$ However, since every vector $x$ of Euclidean
norm~$1$ also satisfies $\|x\|_\infty\le1,$ this would only
be a stronger notion of ``$\eps$-far'' thus applying to fewer matrices.
Note that the expression $\max_{x:\|x\|_\infty\le1} \langle x,Bx\rangle$ is a
$2$-CSP to which we can apply our subsampling theorem (after discretization of
the domain.) This lets us distinguish between matrices that are positive
semidefinite and those that are $\eps$-far from a small subsample. Formally,
we get the following corollary. The simple proof is omitted.

\begin{corollary}
\label{cor:proppsd}
Let $B$ by a matrix with $\|B\|_\infty\le D/n^2.$ Then there is a
property testing algorithm ${\cal A}$ such that:
If $B$ is $\eps$-far from being positive semidefinite, then ${\cal A}$ rejects
$B$ with probability greater than~$2/3.$
If $B$ is
positive semidefinite, then ${\cal A}$ rejects $B$ with probability less than~$1/3$. \Tnote{I added $1/3$ here, it was missing, I hope it should not be another constant.}
Furthermore, ${\cal A}$ reads only $\poly(D,\eps^{-1})$ many
entries of $B$ and runs in time $\poly(D,\eps^{-1})$
\end{corollary}
%
%

\ifnum\full=0
\else

\fi

\section{\maxcut in random geometric graphs}
\label{sec:maxcut}

In this section we discuss the application of our theorem to solving
\maxcut in random geometric graphs.
Let us first recall some basic facts.
The value of the maximum cut of a graph $G$ is given by
$\opt(G)\seteq\max_{x\in\{-1,1\}^n}\langle x,\nfrac 14 L(G)x\rangle$. Here
$L(G)$ denotes the combinatorial Laplacian of $G.$
The Goemans-Williamson~\cite{GoemansWi95} relaxation for \maxcut is
$\sdp(G)=\max\left\{\vbig \nfrac 14 L(G)\bullet X\mid
X\gep0,\forall i\colon X_{ii}=1\right\}
\mper$
Note that $\opt(G)$ and $\sdp(G)$ range between $0$ and $1,$ the total
edge weight of a normalized graph.
We will consider relaxations obtained by adding valid constraints to the above
program. A specific set of constraints we'll be interested in are the
$\ell_2^2$ \emph{triangle inequalities} which can be expressed by adding the
constraint $X_{ij}+X_{jk}-X_{ik}\le1$
and $X_{ij}+X_{jk} + X_{ik} \geq  -1\mper$ for every $i,j,k\in V.$
The relaxation including triangle inequalities will be denoted $\sdp_3(G).$

%
%
%
%
\paragraph{Sphere graphs.}
We denote by $G_\gamma$ the graph on the vertex set
$V=\mathbb{S}^{d-1}$ with edge set
$
E = \{ (u,v)\in V\times V \mid \tfrac14\|u-v\|^2\ge1-\gamma\}\mper
$
%
The integral value of $G_\gamma$, denoted $\opt(G_\gamma),$ is defined as the
maximum of
$\mu(A,\bar A)\defeq\mu^2(\{(x,y)\in E\colon x\in A,y\not\in A\})$
taken over all measurable subsets $A\subseteq\mathbb{S}^{d-1}$
Here, $\mu$ denotes the uniform surface measure of the sphere $\mathbb{S}^{d-1}$
and $\mu^2=\mu\times\mu.$ A theorem of Feige and Schechtman shows that the
maximum is attained for any hemisphere.
\begin{theorem}[Feige-Schechtman~\cite{FeigeSc02}]
Fix $\gamma\in[0,1]$ and consider the graph $G_\gamma.$ Then, the maximum of
$\mu(A,\bar A)$ over all measurable subsets $A\subseteq\mathbb{S}^{d-1}$
is attained for any hemisphere $H\subseteq\mathbb{S}^{d-1}.$
\end{theorem}
Recall, if $A$ is a hemisphere, $\mu(A,\bar A)=1-\Theta(\sqrt{\gamma}).$ Hence
$\opt(G_\gamma)=1-\Theta(\sqrt{\gamma}).$
At this point we mention that the SDP relaxation for \maxcut is
well-defined on infinite graphs though we omit the formal details.
In this case it is easiest to think of $E$ as
a distribution over edges so that the SDP maximizes the quantity
$\E_{(u,v)\sim E}\frac14\|f(u)-f(v)\|^2$ over all embeddings
$f\colon V\to B$ satisfying the usual
additional constraints. Here $B$ can be taken to be the unit ball of
the infinite dimensional Euclidean space.

The sphere graph itself can then be interpreted as an SDP solution, hence the
following fact.
\begin{fact}[Basic SDP]
$\sdpGW(G_\gamma)\ge1-\gamma.$
\end{fact}
\begin{proof}
The graph itself gives an embedding (the identity embedding)
such that for each edge $(u,v)\in E$, $\frac14\|u-v\|^2\ge1-\gamma.$
Since the SDP averages this quantity over all edges
in the graph, the claim follows.
\end{proof}
We will show next that triangle inequalities change the value of the SDP from
$1-\gamma$ to $1-\Omega(\sqrt{\gamma})$ thus capturing the integral value up to
constant factors in front of~$\gamma$.
\torestate{Lemma}{lem:sdp3}{
$\sdpGWT(G_\gamma)\le 1-\Omega(\sqrt{\gamma})\mper$
}
This lemma was quite possibly previously known, but we will give a proof in Section~\ref{sec:maxcut-details} for lack
of a reference.
Using standard discretization arguments all previous lemmas can be transferred
to a sufficiently dense discretization of the continuous sphere. Similarly, it
is not difficult to show that sufficiently many random points from the sphere
will give a good discretization.

\torestate{Lemma}{lem:dense}{
Fix $\gamma\in[0,1], d\in\N.$ Then, there exists an $n_0(d,\gamma)\in\N$ so that if we pick
$V\subseteq S^{d-1}$ uniformly at random with $|V|\ge n_0$, then the induced subgraph $G_\gamma[V]$
satisfies
(1) $\opt(G_\gamma[V])=1-\Theta(\sqrt{\gamma}),$ and
(2) \item $\sdpGWT(G_\gamma[V])=1-\Theta(\sqrt{\gamma}).$
} 
The proof is given in Section~\ref{sec:maxcut-details}.
It is worth noting that the proof of the previous lemma gives a very weak
bound on the number of vertices that we are required to subsample. In
particular, it is not difficult to see that the average degree of the graph
will be $n^{1-o(1)}.$ A priori, it could therefore be the case that the SDP
value changes when considering a subsample of the sphere with average degree
$\log(n)$ or even $O(1).$ Indeed,~\cite{FeigeSc02} show that for some fixed
$\gamma,$ a random subsample of the sphere of expected degree $O(\log n)$ will
satisfy most triangle inequality constraints with high probability thus
exhibiting some integrality gap for $\sdp_3.$\footnote{In their example
the angle between two neighboring vertices is chosen to be more than 60 degrees
corresponding to very large~$\gamma$ to which our theorem does not apply due
to the constant factor loss in~$\gamma.$}
However, our main theorem in this section implies that asymptotically
$\sdp_3$ behaves like $1-\sqrt{\gamma}$ rather than $1-\gamma.$

To argue this, we'd like to use our subsampling theorem for unique games. It is well known how to express the max-cut
problem on a graph $G$ as an instance $\cG$ of \uniquegames where the constraint graph is exactly~$G.$ Since we defined
unique games to be minimization problems, this corresponds to minimizing the number of uncut edges. We therefore have
that $\opt(G)=1-\opt(\cG)$ and furthermore it is well known that $\sdp(G)=1-\sdp(\cG)$ for the basic SDP relaxation and
also $\sdp_3(G)=1-\sdp_3(\cG)$ where the latter refers to an SDP relaxation for \uniquegames that includes triangle
inequalities, yielding the following theorem:

\begin{theorem}\label{thm:sdp3}
Fix $\gamma\in[0,1]$ and let $\Delta>\poly(1/\gamma).$ Fix $d$ and choose $n$ such that for $n$
uniformly random points $V\subseteq\mathbb{S}^{d-1}$ the induced graph
$G_\gamma[V]$ has expected degree $\Delta.$ Then,
\[
\sdpGWT(G_\gamma[V])=1-\Theta(\sqrt{\gamma}).
\]
\end{theorem}
\begin{proof}
We think of $G_\gamma[V]$ as a uniform vertex subsample of a random dense
discretization $G_\gamma[W]$ in $d$ dimension. Note that by
Lemma~\ref{lem:dense} we have $\sdp_3(G_\gamma[W])=1-\Theta(\sqrt{\gamma}).$
We can reduce $G_\gamma[W]$ to a unique game $\cG$ so that
$\sdp_3(\cG)=\Theta(\sqrt{\gamma}).$ Now $G_\gamma[V]$ corresponds to the
unique game $\cG[V],$ since the constraint graph of $\cG[V]$ is precisely
$G_\gamma[V].$ By Theorem~\ref{thm:subsample-ug} (subsampling theorem for
\uniquegames), we know that $\sdp_3(\cG[V])=\Theta(\sqrt{\gamma}).$ Note that
triangle inequalities correspond to a reasonable relaxation. But then it
follows that $\sdp_3(G_\gamma[V])=\Theta(\sqrt{\gamma}).$
\end{proof}

Theorem~\ref{thm:maxcut} is a corollary of this theorem, since $\sdp_3$ can
now be used to certify that random geometric graphs have small max-cut value.

\section{Negative results for subsampling}
\label{sec:negative}

In this section we first observe that its is impossible to obtain even a weak
subsampling result for the semidefinite programming relaxation of $k$-CSPs
with~$k\ge 3.$ This results follows from Schoenebeck's integrality
gap~\cite{Schoenebeck08}.  We also argue that even in the case of $2$-CSPs
subsampling is impossible when the constraints are \emph{not} unique.

Second, we give a separation between semidefinite programming and linear
programming by showing that a subsampling result for linear programming
is impossible even in the case of \maxcut and \uniquegames.
Here, our results are based on the
integrality gap construction of~\cite{CharikarMaMa09}.

\subsection{No subsampling for SDP relaxations of $k$-CSPs with $k\ge 3$.}

\begin{theorem} \label{thm:no-subsample-sdp}
There is a $k$-CSP $\cP$ with $\Omega(n^k)$ constraints in the variables~$[n]$
so that $\sdp_{O(1)}(\cP)\le0.51,$ but with high probability
$\sdp_{\Omega(n)}(\cP[U])\ge0.99$ where $U\subseteq[n]$ is a random set of
size $\delta n$ with $\delta>c/n^{1-\nfrac1k}$ for some constant~$c.$
\end{theorem}

\begin{proof}[Proof sketch.]
We may take $\cP$ to be a random dense instance of $k$-XOR. It is known that
an SDP with a constant number of rounds of Lasserre captures the integral
value of the CSP. Now $\cP[U]$ is a $k$-XOR instance with $\Omega(\delta^kn^k)
=Cn$ constraints for
some constant~$C.$ For large enough $C~,$ the result of~\cite{Schoenebeck08}
then implies the claim.
\end{proof}

\subsection{No subsampling for SDP relaxations of non-unique $2$-CSPs}

The above result also shows that we cannot hope for a subsampling theorem for
semidefinite relaxations of \emph{non-unique} $2$-CSPs. Indeed, we can take a
dense instance $\cP$ of $3$-SAT and express it as a $2$-CSP $\cP'$ as follows:
Every constraint $P\in\cP$ gets mapped to a new variable $x_P$ over the
alphabet $[8].$ Each label represents an assignment to the original constraint.
Every two constraints sharing one variable in $\cP$ contribute one constraint
$P'\in\cP'$ which enforces that the assignment to the shared variable is
consistent.

Subsampling variables in $\cP'$ corresponds to subsampling constraints in
$\cP.$ Using~\cite{Schoenebeck08}, the subsample of $\cP$ will be a gap
instance for the Lasserre hierarchy. Since our reduction is local, ideas
of~\cite{Tulsiani09} show that also the subsample of $\cP'$ will be a gap
instance. This rules out the possibility of a subsampling theorem for
non-unique $2$-CSPs of alphabet size $8.$

\subsection{No subsampling for LP relaxations of $2$-CSPs}
In this section we rule out subsampling theorems for strong linear programming relaxations even in the case of \maxcut
for which strong semidefinite relaxations do admit a subsampling theorem. Specifically, we consider the Sherali-Adams
LP relaxation for \maxcut: $\sa_r(G)= \max \sum_{(u,v)\in E}x_{uv}$ over $(u,v)$ s.t. the vector $(x_{uv})_{u,v\in V}$
lies in the Sherali-Adams relaxation of the cut polytope.

 The Sherali-Adams relaxation of the cut
polytope is obtained by applying $r$ rounds of lift-and-project operations to the base set of linear inequalities that
define the metric polytope, i.e., $\{ x_{ij}+x_{jk}\ge x_{ik},x_{ij}+x_{jk}+x_{ik}\le 2, x_{ij}=x_{ji},1\ge x_{ij}\ge
0\}.$ For a formal definition  see, for instance, \cite{CharikarMaMa09}.

The next theorem shows that there are graphs which have Sherali-Adams value
bounded away from~$1$ for a constant number of rounds. But after subsampling
the value comes arbitrarily close to~$1$ even when considering a huge number
of rounds.
\begin{theorem}\label{thm:sa}
For every function $\e=\e(n)$ that tends to $0$ with $n$, there exists a function $r=r(n)$ that tends to
$\infty$ with $n$ and family of graphs $\{ G_n \}$ of degree $D=D(n)$  such that
\begin{enumerate}
\item For every $n$, $\mathrm{lp}_3(G_n) \leq 0.8$
\item If $G'$ is a random subgraph of $G$ of size $(n/D)^{1+\e(n)}$ then $\E[ \mathrm{lp}_{r(n)}(G')] \geq 1 -
    \frac{1}{r(n)}$.
\end{enumerate}
where $\mathrm{lp}_k(H)$ denotes the value of $k$ levels of the Sherali-Adams
linear program for Max-Cut on the graph $H$.
\end{theorem}

\begin{proof}[Proof sketch] Let $G_n=G_{n,p}$ for some $p\le\frac12.$
It is not difficult to argue that three rounds of Sherali-Adams have value at most $0.7$ on $G=G_n$ with high
probability over $G_n$ itself. This follows by considering triangles in $G$ and arguing that every edge in $G$ occurs
in the same number of triangles up to negligible deviation. But $3$ rounds of Sherali-Adams have value at most $2/3$ on
a triangle. Hence, $\sa_3(G)\le\nfrac23+o(1).$


On the other hand let $\delta=\frac{n^{\eps}}{D}$ where $D=pn$ is the expected
degree of $G.$ We observe that $G'=G[V_\delta]$ is exactly distributed like
$G'=G_{m,\lambda/m}$ for $m=(n/D)^{1+\eps}$ and $\lambda= m^\eps.$
Using arguments
similar to~\cite{AroraBoLoTo06},
one can check that such graphs have
  girth going to infinity,
  and for some $M \in \omega(1)$, all subsets size $M$ are
  $(1+\eta)$-sparse, where $\eta \in o(1)$.
  Hence, we can follow the proof as above and
  use~\cite{CharikarMaMa09} to argue that $G[V_\delta]$ has
  Sherali-Adams value larger than $1-o(1)$ for $\omega(1)$ rounds, and
  therefore picking $r(n)$ sufficiently small concludes the proof sketch.
\end{proof}
\begin{remark}\label{rem:anygraph}
We remark that such expansion based arguments can be used to give similar
results for subsamples of any~$\Delta$-regular graph and in particular for
subsamples of the Feige-Schechtman graph.
\end{remark}

\section{Proof of the main theorem for Unique Games}
\label{sec:subsample-ug}
%

\newcommand{\sge}{\succeq}
\newcommand{\sle}{\preceq}
\newcommand{\bdot}{\bullet}

We now come to the proof of our main theorem ---a weak subsampling theorem for
strong SDP relaxations of \uniquegames.
Let us first formalize the notion of a ``reasonable'' SDP relaxation.
\begin{definition}[Reasonable SDP relaxation for \uniquegames]
  \label{def:reasonable}
  Let $V$ be a set of  $n$ vertices and let $\cM$
  be a convex subset of the set $\cM_1$ defined as
  \begin{displaymath}
    \cM_1
    \defeq \Set{
      X \in \R^{(V\times[R])\times (V\times[R])}
      \mid X\sge 0, %
      \quad \forall i\in V,a\in[R].~~ X_{ia,ia} \le 1} %
    \mper
  \end{displaymath}
  \Tnote{Changed $q$ to $R$}
  For a unique game on a graph $G$ with vertex set $V$, we define
$\sdp_\cM(\cG)$ by
  \begin{displaymath}
    \sdp_\cM(\cG)
    \defeq \min_{X\in \cM}
\E_{(u,v)\in E} \E_{a\in[R]} \|u_a - v_{\pi_{v\la u}(a)}\|^2
    \mper
  \end{displaymath}
  We say that $\sdp_\cM$ is a \emph{reasonable} relaxation for
  \uniquegames if $\cM$ is closed under renaming of
  coordinates and permutation of labels in the sense that
  \begin{displaymath}
    \forall F\from V\to V.
    \quad \forall \pi_{1},\ldots,\pi_{n}\from [R]\to[R].
    \quad \forall X\in \cM.
    \quad (X_{F(i)\pi_{i}(a),F(j)\pi_{j}(b)})_{i,j\in V,a,b\in[R]} \in \cM
    \mper
  \end{displaymath}
  \Dnote{the decoding might require that we rename alphabet symbols
    (apply the UG permutation). I think for this kind of renaming it
    is important that it is bijective. Right now, I am thinking we
    want permutations $\pi_1,\ldots,\pi_n$ of $[q]$ and then consider
    the matrix $(X_{F(i)\pi_i(a),F(j)\pi_j(b)})$. Not sure though.
  }
  \Tnote{You are right, and I changed the definition accordingly}
  Here, the function $F$ is not required to be bijective, but for every
$u,v\in V,$ $\pi_{v\la u}$ is a permutation of $[R].$
  We also say that $\cM$ is \emph{reasonable} if it satisfies the
  condition above.
\end{definition}
\Mnote{Why is this definition reasonable? Maybe this needs some
  justification.}  \Tnote{Tried an explanation in the following} For
an SDP to be reasonable it is only needed that any set of vectors used
for one vertex of the unique game can also be used in any other
vertex, even after a permutation of the labels.

%
The next theorem gives a subsampling result for any reasonable relaxation of \uniquegames.
\begin{theorem}[Main]
  \label{thm:subsample-ug}
  Let $\e >0$ and let $\cG$ be a unique game
on a $\Delta$-regular constraint graph.
  Then, for $\delta=\Delta^{-1}\cdot \poly(\nfrac 1\e)$,
  \begin{displaymath}
    \tfrac19\sdp_\cM(\cG) - \e 
    \le \E \sdp_\cM(\cG[V_\delta])
    \le \sdp_\cM(\cG) + \e 
    \mcom
  \end{displaymath}
  where $\sdp_\cM$ is any reasonable relaxation.
\end{theorem}

The theorem is proven in the next two steps.

\subsection{First step: proxy graph theorem via subsampling theorem}

For our first step we'll need a special case of our subsampling theorem for
semidefinite programs. It shows that under certain regularity conditions
subsampling is possible for semidefinite programs that correspond roughly to
the SDP of a unique game on a regular graph.
\begin{lemma}
  \label{lem:subsample-reformulated}
  Let $\e>0$ and let $\cP$ be a $2$-CSP over $n$ variables
of the form $\cP(x)= \sum_{i,j\in V} b_{ij}P(x_i,x_j)$ where we interpret each variable $x_i$ as a collection of
vectors $x_i=(v_{i,a})_{a\in[R]}$ and each pay off function is bounded and of the form
$P(x_i,x_j)=\sum_{a,b}d_{a,b}\langle x_{i,a},x_{j,b}\rangle\mper$ Assume that each $b_{ij}\le1/\Delta n$ and
$\sum_{i=1}^nb_{ij}=\Theta(1/n)$ for every $j,$ Then, for $\delta\ge\poly(1/\eps)/\Delta,$
  \begin{displaymath}
    \E \opt(\cP[V_\delta])
     = \opt(\cP)
     \pm \e
     \mper
  \end{displaymath}
\end{lemma}
As shown in Section~\ref{sec:subsample-kcsp} this lemma can be derived easily
from Theorem~\ref{thm:subsample-kcsp-1}.
We'll proceed to state and prove our proxy theorem.
\begin{theorem}[Proxy Theorem]
  \label{thm:proxy}
Let $\cG,\cH$ be unique games on $\Delta$-dense constraint graphs and suppose
\begin{equation}\label{eq:proxy-assumption}
\sdp(\cG)[X]\ge c\cdot \sdp(\cH)[X]
\end{equation}
for every SDP solution $X\in\cM_1.$
Then for $\delta\ge\Delta^{-1}\poly(1/\eps),$ we have
\begin{displaymath}
\E \sdp_\cM(\cG[V_\delta]) \ge c\cdot \E \sdp_\cM(\cH[V_\delta])-\eps
  \mper
\end{displaymath}
\end{theorem}
\begin{proof}[Proof of Theorem~\ref{thm:proxy}]
Consider the $2$-CSP instance,
\[
\max_{X\in\cM_1}\cP(X) = c\cdot \sdp(\cH)[X]-\sdp(\cG)[X].
\]
Note that by our assumption
\[
\opt(\cP)=\max_{X\in\cM_1}\cP(X)\le0\mper
\]
Let $A(G)$ and $A(H)$ denote the adjacency matrices of $G$ and $H$ respectively. Since $G$ is $\Delta$-regular and $H$
has degree at least $\Delta,$ we know that each entry of $B=cA(H)-A(G)$ is bounded by $O(1/\Delta n),$ whereas each
row/column in $B$ sums up to $\Theta(1).$ Hence, the matrix~$B$ satisfies the assumption of
Lemma~\ref{lem:subsample-reformulated}. It remains to check that in $\cP$ each pay off function is bounded. This
follows from the fact that in both $\sdp(\cH)$ and $\sdp(\cG)$ each pay-off function is of the form
$\E_{a\in[R]}\|u_a-v_{\pi(v)}\|^2$ and this expression is bounded since each vector has norm at most~$1$ so that each
payoff function is bounded by $O(R).$

Therefore, by Lemma~\ref{lem:subsample-reformulated},
\begin{align*}
\eps \ge \E\opt(\cP[V_\delta])
& =\E\max_{X\in\cM_1} c\sdp(\cH[V_\delta])[X]-\sdp(\cG[V_\delta])[X]\\
& \ge\E\max_{X\in\cM} c\sdp(\cH[V_\delta])[X]-\sdp(\cG[V_\delta])[X]
\tag{since $\cM\subseteq\cM_1$}\\
& \ge\E\max_{X\in\cM}
c\sdp(\cH[V_\delta])[X]-\max_{X\in\cM}\sdp(\cG[V_\delta])[X]\mper
\end{align*}
Hence,
\begin{displaymath}
\E \sdp_\cM(\cG[V_\delta]) \ge c\cdot \E \sdp_\cM(\cH[V_\delta])-\eps
  \mper
\end{displaymath}
\end{proof}

\subsection{Second step: proxy graphs for unique games}

In this section, we show that taking the ``third power'' of a unique game results in a useful proxy graph.
\begin{definition}[Third power of a unique game]
\label{def:ug-power} For a unique game $G$ we define $G^3$ to be the unique game defined on the third graph power of
the constraint graph. An edge $e=(u,v)$ therefore corresponds to a path $(u,u',v',v)$ in $G.$ The constraint $\pi_{v\la
u}$ on the edge $(u,v)$ is defined as the composition of the three constraints along the path in $G,$ that is
\begin{equation}\label{eq:compose}
\pi_{v\la u} = \pi_{v\la v'} \circ \pi_{v'\la u'} \circ \pi_{u'\la u}\mper
\end{equation}
\end{definition}
\begin{lemma}
  \label{lem:second-third-step-1}
Let $G$ denote a $\Delta$-regular unique game.
  Then, for every SDP solution $X\in\cM_1,$
  \begin{displaymath}
\sdp(G)[X]
    \ge \nfrac19 \cdot \sdp(G^3)[X]
    \mper
  \end{displaymath}
\end{lemma}
\begin{proof}
Let $X\in\cM_1$ and let $(u,v)$ be an edge in $G^3$ corresponding to a path $(u,u',v',v)$ in $G.$ Let $a\in[R]$ and put
$a'=\pi_{u'\la u}(a),$ $b'=\pi_{v'\la u'}(a'),$ \Tnote{Changed this from $b' = \pi_{v' \la v}(a')$} and $b=\pi_{v\la v'}(b').$ \Tnote{Changed this from $b=\pi_{v\la v}(b').$} Note that by definition of $G^3,$ we have
$\pi_{v\la u}(a)=b.$ By triangle inequalities
\[
\|u_a - v_b\| \le
\|u_a - u'_{a'}\| + \|u'_{a'}-v'_{b'}\| + \|v'_{b'}-v_b\|
\]
Squaring both sides and taking expectation over $a\in[R],$ we get
\[
\E_{a\in[k]}\|u_a - v_b\|^2 \le
3\E_{a\in[k]} \|u_a - u'_{a'}\|^2 +3\E_{a\in[k]} \|u'_{a'}-v'_{b'}\|^2
+3\E_{a\in[k]} \|v'_{b'}-v_b\|^2\mper
\]
Averaging over edges in $G^3,$ we get
\[
\E_{(u,v)\in G^3}
\E_{a\in[k]}\|u_a - v_{\pi_{v\la u}(a)}\|^2 \le
9\E_{(u,v)\in E}\E_{a\in[k]}\|u_a - v_{\pi_{v\la u}(a)}\|^2\mper
\]
\end{proof}

\begin{lemma}
  \label{lem:second-third-step-2}
  Let $\cG$ be a $\Delta$-regular unique game on a graph $G=(V,E)$
  and let $\widetilde\cG$ be the unique game on a graph $\widetilde
G=(V_\delta,\widetilde E)$ defined by the edge distribution
  \begin{itemize}
  \item sample a random edge $(i,j)$ from $G$,
  \item choose $u$ and $v$ to be random neighbors of $i$ and $j$ in $V_\delta$ (if $i$ or $j$ have no neighbor in
      $V_\delta$, choose a random vertex in $V_\delta$),
  \item output $(u,v)$ as an edge in $\widetilde E$. The constraint on the edge $(u,v)$ is taken to be the
      composition of the constraints on $(u,i),(i,j),(j,v)$ the same way as in Definition~\ref{def:ug-power}.
  \end{itemize}
  Then for $\delta>\Delta^{-1}\poly(\nfrac1 \e)$,
  \begin{displaymath}
    \E \normtv{G^3[V_\delta] - \widetilde G}
    \le \e
    \mper
  \end{displaymath}
Here, $\normtv{\cdot}$ denotes statistical distance.
\end{lemma}
\begin{proof}[Proof Sketch]
  If every vertex of $G$ has the same number of neighbors in
  $V_\delta$, then the two graphs $G^3[V_\delta]$ and $G'$ are
  identical.
  For $\delta>\Delta^{-1}\poly(1/\e)$, the following event happens
  with probability $1-\e$:
  Most vertices of $G$ (all but an $\e$ fraction) have up to a
  multiplicative $(1\pm \e)$ error the same number of neighbors in
  $V_\delta$.
  Conditioned on this event, it is possible to bound
  $\normtv{G^3[V_\delta] - G'}$  by $O(\e )$.
  Assuming this fact, the lemma follows. The details can be found in
Section~\ref{sec:proxy}.
\end{proof}


\begin{lemma}
\label{lem:second-third-step-3}
  Let $\cG$ be unique game on a $\Delta$-regular constraint graph.
  Then for $\delta>\Delta^{-1}\cdot \poly(\nfrac 1\e)$
  and for any reasonable relaxation $\sdp_\cM$,
  \begin{displaymath}
    \E \sdp_\cM(\cG^3[V_\delta])
    \ge \sdp_\cM(\cG) -\e
    \mper
  \end{displaymath}
\end{lemma}

\begin{proof}
  Let $\widetilde\cG$ be as in Lemma~\ref{lem:second-third-step-2} and let
$\widetilde X$
  be an optimal solution for $\sdp_\cM(\widetilde\cG).$

  Let $\cF(V_\delta)$ denote the distribution over mappings $F\from
  V\to V_\delta$, where for every vertex $i\in V$, we choose $F(i)$ to
  be a random neighbor of $i$ in $V_\delta$ (and if $i$ has no
  neighbor in $V_\delta$, we choose $F(i)$ to be a random vertex in
  $V_\delta$).
  For convenience, we introduce the notation $N(i, V_\delta)$ for the
  set of neighbors of $i$ in $V_\delta$ (if $i$ has no neighbor in
  $V_\delta$, we put $N(i,V_\delta)=V_\delta$).
For each $F\sim\cF(V_\delta)$ we define a \emph{decoded} SDP solution $\cA_F(\widetilde X)$ for $\cG.$ Specifically,
the entry corresponding to $i,j\in V$ and labels $a',b'\in[R]$
\begin{enumerate}
\item  Let $F(i)=u$ and $F(j)=v.$ Assigning the label $a'$ to $i$ forces $j$ to have label $b'=\pi_{j\la i}(a')$
    and hence $u$ and $v$ must have labels $a=\pi_{u\la i}(a')$ and $b=\pi_{v\la j}(b').$
\item Define
\[
\cA_F(\widetilde X)_{ia',jb'}:=\widetilde X_{ua,vb}
= \widetilde X_{F(i)\pi_{u\la i}(a),F'(j)\pi_{v\la j}(\pi_{j\la i}(a))}\mper
\]
\end{enumerate}
  Since $\cM$ is reasonable (see Definition~\ref{def:reasonable}), we
  have $\cA_F(\widetilde X) \in \cM$ for any mapping $F\from V\to V_\delta$.

  We define $\cA_{\cF(V_\delta)}(\widetilde X) \seteq \E _{F\sim \cF(V_\delta)}
  \cA_F(\widetilde X)$.
  Since $\cM$ is convex, we also have
  \begin{displaymath}
    \cA_{\cF(V_\delta)}(\widetilde X)
    =\E_{F\sim \cF(V_\delta)} \cA_F(\widetilde X)
    \in \cM
    \mper
  \end{displaymath}
  We claim that
\[
\sdp(\cG)[\cA_{\cF(V_\delta)}(\widetilde X)] = \sdp(\widetilde\cG)[\widetilde X]\mper
\]
Indeed,
  \begin{align}
  \sdp(\cP)[\cA_{\cF(V_\delta)}(\widetilde X)]
    & =2 \E_{ij\sim G}\E_{a'\in[R]}
    \cA_{\cF(V_\delta)}(\widetilde X)_{ia',j\pi_{j\la i}(a')} \notag \\
    & =2 \E_{ij\sim G}\E_{a'\in[R]} \E_{F\sim \cF(V_\delta)} %
 \widetilde
X_{F(i)\pi_{F(i)\la i}(a'),F'(j)\pi_{F(j)\la j}(\pi_{j\la i}(a'))}
\notag\\
& = 2 \E_{u'v'\sim G}~ \E_{a\in[R]}
\E_{u\in N(i,V_\delta)} ~\E_{v\in N(j,V_\delta)}
\widetilde X_{u\pi_{u\la i}(a'),v\pi_{v\la j}(\pi_{j\la i}(a'))}
\notag\\
& = 2 \E_{u'v'\sim G}~ \E_{a\in[R]}
\E_{u\in N(i,V_\delta)} ~\E_{v\in N(j,V_\delta)}
    \widetilde X_{ua,v\pi_{v\la u}(a)} \tag{using~(\ref{eq:compose})}\notag    \\
& = 2 \E_{u v \sim \widetilde G}\E_{a\in[R]}
\widetilde X_{ua,v\pi_{v\la u}(a)} \notag\\
    \label{eq:second-third-step-1}
&= \sdp(\widetilde \cG)[\widetilde X]\mper
  \end{align}
  It follows that
  \begin{align}
    \sdp_\cM(\cG) %
    & \le
  \sdp_\cM(\cG)[\cA_{\cF(V_\delta)}(\widetilde X)]
    \tag{using $\cA_{\cF(V_\delta)}(\widetilde X)\in \cM$}
    \notag\\
&= \sdp_\cM(\widetilde \cG)[\widetilde X]
    \qquad\using{\eqref{eq:second-third-step-1}}
    \notag    \\
    & = \sdp_\cM(\cG')\mper
    \label{eq:second-third-step-2}
  \end{align}
  We can now finish the proof of the lemma,
  \begin{align*}
    \E \sdp_\cM(\cG^3[V_\delta]) %
    &\ge \E \sdp_\cM(\widetilde \cG) - O(1)\E \normtv{G^3[V_\delta]-G'} \\
    &\ge \sdp_\cM(\cG) -\e  %
    \qquad\using{\eqref{eq:second-third-step-2} and
      Lemma~\ref{lem:second-third-step-2}} \mper
  \end{align*}
\end{proof}

\subsection{Putting things together}
By combining the previous two steps we can prove Theorem~\ref{thm:subsample-ug}.

\begin{proof}[Proof of Theorem~\ref{thm:subsample-ug}]
  We need to show that
  \begin{displaymath}
    \tfrac19 \sdp_\cM(\cG) - O(\e)
    \le \E \sdp_\cM(\cG[V_\delta])
    \le \sdp_\cM(\cG)
    \mper
  \end{displaymath}
  The upper bound on $ \E \sdp_\cM(\cG[V_\delta]) $ is easy to show.
  We consider an optimal solution $X\in \cM$ for $\cG.$
  Note that the value of $X$ is preserved for
  $\cG[{V_\delta}]$ in expectation, i.e.,
  \begin{displaymath}
\E \sdp_\cM(\cG[{V_\delta}]) \le
    \E \sdp_\cM(\cG[{V_\delta}])[X]
    =  \sdp_\cM(\cG)[X] \mper
  \end{displaymath}
\Mnote{we used to have $\pm\eps$ here, don't know why}

  We combine the lemmas in this section to prove the lower bound.
  Notice that by Lemma~\ref{lem:second-third-step-1}, we can choose
  $\cH=\cG^3$ (for $c=1/9$) in Lemma~\ref{thm:proxy}.
  With this choice of $\cH$, we can finish the proof of the theorem,
  \begin{align*}
    \E \sdp_\cM(\cG[V_\delta])
    &\ge \tfrac19\cdot \E \sdp_\cM((\cG^3)[V_\delta])
    - \e
    \qquad\using{Lemma~\ref{thm:proxy}}
    \\
    & \ge \tfrac19\cdot \sdp_\cM(\cG)
    - \tfrac {10}9 \e
    \qquad\using{Lemma~\ref{lem:second-third-step-3}}
    \mper
  \end{align*}
\end{proof}

\section{Proof of subsampling theorem}
\label{sec:proof-kcsp}

In this section prove our main subsampling theorem for $k$-CSPs. 
We will work with the following notion of density.
\begin{definition}[density]
We say that a $k$-CSP $\cP$ is $\Delta$-dense if
$|P|\le1$ for every $P\in\cP$ and furthermore
for every $r\in[k]$ and fixing of $k-1$ variables
$I=(i_1,\dots,i_{r-1},*,i_{r+1},\dots,i_k),$ we have
\begin{equation}
\frac\Delta c\le\sum_{P\in\cP\colon\var(P)=I} |P|\le c\Delta
\end{equation}
for some absolute constant $c>0.$
\end{definition}
Here, $\Delta$ is a parameter in $[1,n].$ The larger $\Delta$ the denser the
instance. In a $2$-CSP $\Delta$ corresponds to the degree of each variable.
In a \emph{dense} $k$-CSP~\cite{AlonVeKaKa03} we have $\Delta=\Theta(n).$
%
\begin{theorem}\label{thm:subsample-kcsp}
Let $\eps>0,\Delta>1.$
Let $\cP$ be a $\beta$-dense $k$-CSP in $n$ variables over
an alphabet of size $q.$ Put $\delta\ge\eps^{-C}\log(q)/\Delta$ for
some absolute constant $C.$
Suppose $U\subseteq[n]$ is chosen uniformly at
random so that $|U|=\delta n.$ Then,
\begin{equation}
\left|
\E\opt(\cP_U)
-\opt(\cP)
\right|\le\epsilon\mper
\end{equation}
\end{theorem}
\begin{remark}
In the case $k=2,$ our notion of density reduces to the usual notion of
density in a graph. We get the optimal trade-off between density and sample
size in that case. When $k>3$ there are $k$-CSPs with $n^{k-1}$ constraints
that do not allow sparsification. For instance, consider a dense $k$-CSP in
which all constraints share the same variable. We cannot subsample here, since 
we would likely lose that variable and hence all constraints.
\end{remark}
Throughout the proof we will think of $k$ as an absolute constant
and consider any function of $k$ as~$O(1).$ We will
also assume that coordinates in $[n]$ are
sampled \emph{with} replacement.


One direction of the theorem is immediate.
\begin{lemma}\label{lem:onedirection}
\begin{equation}
\Ex{\max_{x\in[q]^n} \cP_U(x)}
\ge \max_{x\in[q]^n}\cP(x)\mper
\end{equation}
\end{lemma}
\begin{proof}
Suppose $x^*\in[q]^n$ maximizes $\cP(x).$ Note that $\Ex{\max_{x\in[q]^n}\cP_U(x)}\ge\Ex{\cP_U(x^*)}.$ On the other
hand, $\E\cP_U(x)=\cP(x).$
\end{proof}

The other direction requires all the work. We will split it up into two main
lemmas.  The first lemma shows that the subsampling step is random enough to
give a concentration bound for large subsets of $[q]^n.$
\begin{lemma}[Concentration]
\label{lem:randomness} 
\Mnote{fix parameters}
There are constants $c_0,c_1$ so that for  
$\delta_0=\eps^{-c_0}\log(q)/\Delta$ and
randomly chosen $U\subseteq[n]$ of size $|U|\ge\delta_0n$ 
we have that for every subset 
$\Psi\subseteq[q]^n$ of size $|\Psi|\le\exp(\eps^{c_1}|U|),$ 
\begin{equation}
\left|
\Ex{\max_{x\in \Psi} \cP_U(x)}
-\max_{x\in\Psi}\cP(x) 
\right|
\le \epsilon \mper
\end{equation}
\end{lemma}
We think of $\delta_0 n$ as the smallest sample size for which we can expect
concentration. The previous lemma shows that the maximum value of any fixed
set of $\exp(\poly(\eps)|U|)$ assignments is preserved when sampling $U$ of 
size larger than $\delta_0n.$

The second main lemma shows that this concentration bound is actually good
enough for us. Indeed, the maximum of the subsample turns out to have enough
redundancy so that we can find a suitably small set of assignments in $[q]^n$
that captures the optimal value of the subsample up to a small error.
\begin{lemma}[Structure]
\label{lem:structure} 
For every constant $c$ there is a constant~$C$ and a
set of assignments~$\Psi\sse[q]^n$ of size 
$|\Psi|\le\exp(\eps^c\delta n)$ where
$\delta=\eps^{-C}\log(q)/\Delta$ such that for randomly chosen $U$ of 
size $|U|=\delta n,$ we have 
\begin{equation}
\Ex{\max_{x\in[q]^n}\cP_U(x)}
\le\Ex{\max_{x\in \Psi}\cP_U(x)}+\eps  \mper
\end{equation}
\end{lemma}

Together these two lemmas direcly imply the main subsampling theorem as shown
next.
\begin{proof}[Proof of Theorem~\ref{thm:subsample-kcsp}]
In one direction, let $\Psi$ be the set from Lemma~\ref{lem:structure} which
we obtain for $c=c_1$ where $c_1$ is the constant from
Lemma~\ref{lem:randomness}. Let $C$ be the constant given by
Lemma~\ref{lem:structure} for the given choice of $c.$ Then with
$\delta=\eps^{-C}\log(q)/\Delta,$ we have
\begin{align*}
\Ex{\max_{x\in[q]^n}\cP_U(x)}
&\le \Ex{\max_{x\in\Psi}\cP_U(x)}+\nfrac\eps2 \tag{by Lemma~\ref{lem:structure}}\\
&\le \Ex{\max_{x\in\Psi}\cP(x)}+\eps \tag{by Lemma~\ref{lem:randomness}}\\
&\le \Ex{\max_{x\in[q]^n}\cP(x)}+\eps \mper
\end{align*}
The other direction follows from Lemma~\ref{lem:onedirection}.
\end{proof}

\subsection{Proof of Concentration Lemma}
\newcommand\trc{\ensuremath{^{^{_{\rm prune}}}}}
Fix a vector $x\in[q]^n.$ We will first analyze the case where we 
sample sets $U_1,U_2,\dots,U_k\subseteq[n]$ independently at random of size
$\delta_1 n$ ($\delta_1$ is some parameter that we'll instantiate later) and we
keep all constraints whose $r$th variable is contained in $U_r.$ Later we will
be able to conclude the case where $U_1=U_2=\dots=U_k.$

The argument proceeds in $k$ steps. At each step~$i$ we restrict the $\cP$ to
those constraints whose $r$-th variable is contained in $U_r.$ After each step
we perform a pruning operation in which we remove variables whose
\emph{influence} has become too large. We then argue that
the pruned CSP has the desired concentration properties and moreover that 
pruning doesn't remove to many constraints in expectation.

Denote by $\cP_1$ the CSP obtained from $\cP$ by throwing away all predicates 
whose first variable is not in $U_1.$ Then normalize by a factor $\delta_1^{-1},$ 
since we expect to remove a $\delta_1$ fraction of the predicates. Formally,
\[
\cP_1 = \{ \delta_1^{-1}P \mid \var_1(P)\in U_1, P\in\cP\}\mper
\]
Now, let $\Inf_i(\cP_1)=\frac1{|\cP|}\sum_{P\in\cP_1\colon i\in\var(P)}|P|$ 
denote the \emph{influence} of variable $i$ in $\cP_1.$ Let $N=\E\Inf_i(\cP_1)
=O(1/n).$ 
As mentioned before, we will throw away all predicates that contain a variable 
whose influence in~$\cP_1$ has become too large, say, larger than~$2 N,$
\[
\cP_1\trc = \{ P\in\cP_1 \mid \forall i\in\var(P)\colon \Inf_i(\cP_1)\le 2N\}\mper
\]
We think of this as the pruning of $\cP_1.$ Continue this process,
inductively, by putting
\[
\cP_r = \{ \delta_1^{-1}P \mid \var_r(P)\in U_r, P\in\cP_{r-1}\trc\}\mcom
\]
and
\[
\cP_r\trc=\{ P\in\cP_r \mid \forall i\in\var{P}\colon\Inf_i(\cP_r)\le
2^rN\}\mper
\]
Here $\Inf_i(\cP_r)=\frac1{|\cP|}\sum_{P\in\cP_r\colon i\in\var(P)}|P|.$
Note that $\cP_{k-1}\trc$ will still have maximum influence at
most $O(1/n).$

In the following, 
when we write $\cP_r(x)$ we think of it as normalized in the same way we 
normalize $\cP_U(x),$ i.e., by a factor $1/|\cP|.$
\begin{lemma}
\label{lem:azuma-csp}
For every $x\in[q]^n,$
\begin{equation}
\Pr\left\{
\left|\cP_k(x)
- \cP(x)\right| + t\eps\right\}
\le O(\exp(-t^2\eps^2\delta_1 n))
\end{equation}
\end{lemma}
\begin{proof}
The proof proceeds in $k$ steps. At each step we will apply a variant of
Azuma's inequality (sometimes referred to as McDiarmid's inequality) given by
Lemma~\ref{lem:mcd}.

Specifically, we claim that
\begin{equation}
\label{eq:azuma-r}
\Pr\left\{
\cP_r(x)
> \cP_{r-1}\trc(x) + t\eps \right\}
\le O(\exp(-t^2\eps^2\delta_1 n))\mcom
\end{equation}
for every $0<r\le k.$ We define the mapping
\[
f_r(U_r)=\sum_{P\in\cP_{r-1}\trc,\var_r(P)\in U_r} \delta_1^{-1}P(x) =\cP_r
\]
where we think of $U_r$ as a tuple $(i_1,\dots,i_{\delta_1 n})$ each coordinate being an index in~$[n].$ Note that
\[
\E f_r = \sum_{P\in\cP_{r-1}\trc} \delta_1\delta_1^{-1}P(x)
= \cP_{r-1}\trc(x)\mper
\]
We claim that $f_r$ has Lipschitz constant $O(1/\delta_1 n)$
in the sense that replacing any coordinate $i\in U_r$ by a $i'\in[n]$ can
change the function value by at most~$O(1/\delta_1 n).$ This is because
the influence of
each variable in $\cP_{r-1}\trc$ is at most $O(1/n).$

Lemma~\ref{lem:mcd} then implies
\[
\Pr\left\{f_r>\E f_r+t\eps\right\}\le
\exp\left(-\frac{\Omega(\delta_1^2n^2t^2\eps^2)}{\delta_1 n}\right)
=\exp\left(-\Omega(t^2\eps^2\delta_1 n)\right) \mper
\]
This is what we claimed in~(\ref{eq:azuma-r}).
By a union bound,~(\ref{eq:azuma-r}) holds for all $r\in[k].$ Hence, we can chain 
these inequalities together and get
\[
\Pr\left\{
\left|\cP_k(x)- \cP_1\right| > t\eps\right\}
\le O(\exp(-t^2\eps^2\delta_1 n))\mper
\]
\end{proof}
We'd like to argue that at every pruning step only a few predicates get
removed and hence $\cP_r\trc$ and $\cP_r$ are close. Specifically we'd like to
show that the influence of~$i$ has enough concentration so that it is larger
than twice its expectation only with small probability. This directly gives us
a bound on the expected amount of pruning. The key observation is the next
lemma which shows that the degree of each fixing~$I$ of $k-1$ variables is
concentrated.
\begin{lemma}\label{lem:degree}
Assume $\delta_1\ge1/\eps^2\Delta,$
fix $I=(i_1,\dots,i_{r-1},*,i_{r+1},\dots,i_k)$ and let
$\cQ=\{P\in\cP\mid\var(P)=I\}.$ Then, 
\begin{equation}
\E
\left|
\delta_1\Delta - 
\sum_{P\in\cQ\colon\var_r(P)\in U_r} |P|
\right|\le \eps\delta_1\Delta\mper
\end{equation}
\end{lemma}

\begin{proof}
By the density condition on~$\cP,$ we have
$\sum_{P\in\cQ}|P|\ge\Omega(\Delta)$ and $|P|\le1.$ (In particular,
$|\cQ|\ge\Omega(\Delta).$)

Consider the random variable
\[
Z =\sum_{P\in\cQ,\var_r(P)\in U_r} |P|\mcom
\]
which sums the norm of all predicates in $\cQ$ that are selected by~$U_r.$ 
Let $\mu=\E Z =\delta_1\Delta.$
We can express $Z$ as a sum of independent variables $Z=\sum_{i=1}^{|U|}
Z_i, $ where $Z_i$ is the outcome of the $i$-th sample in $U.$  Since
we sampled with replacement, the $Z_i$'s are independent and
identically distributed. Every $Z_i$ assumes each
value $|P|$ for $P\in\cQ$ with probability $1/n.$ We note that $\E
Z_i=\frac1n\sum_{P\in\cQ}|P|=\Theta(\frac\Delta n).$
Let us compute the fourth moment of
$Z-\E Z.$ First observe that $\E(Z_i-\E Z_i)^4\le O(\E|Z_i|^4)$ and
\[
\E |Z_i|^4 \le \frac1n\sum_{P\in\cQ} |P|^4
\le \frac{O(\Delta)}n \mper
\]
Similarly, for any $i \neq k$: \Tnote{Added $i \neq k$}
\[
\E(Z_i-\E Z_i)^2(Z_k-\E Z_k)^2\le
\frac {O(1)}{n^2}\sum_{P,P'\in\cQ}|P|^2|P'|^2
\le \frac{O(\Delta^2)}{n^2}\mper
\]
By independence and the fact that $\E(Z_i-\E Z_i)=0,$ we therefore have
\begin{align*}
\E(Z-\E Z)^4
&=\sum_{i}\E(Z_i-\E Z_i)^4
+ 6\sum_{i\neq k}\E(Z_i-\E Z_i)^2\E(Z_k-\E Z_k)^2\\
& \le \delta_1 n \cdot \frac{O(\Delta)}n
+
(\delta_1 n)^2\cdot \frac{O(\Delta^2)}{n^2}\\
& =O(\mu^2)
\mper 
\end{align*}
Thus, by Markov's inequality,
\begin{equation}\label{eq:fourth}
\Pr(|Z-\E Z| > t)\le
\frac{E(Z-\E Z)^2}{t^4}
\le \frac{O(\mu^2)}{t^4}
\mper
\end{equation}

Therefore we can bound $\E|Z-\E Z|$ in expectation
by integrating~(\ref{eq:fourth}) over $t\ge 1,$
\begin{equation}
\int_{t\ge 1} t\cdot\Pr(|Z-\E Z|>t\eps\mu)\mathrm{d}t
\le \int_{t\ge 1} t\cdot \frac{O(\mu^2)}{(t\eps\mu)^4}\mathrm{d}t
\le \frac{O(1)}{\eps^4\mu^2}\int_{t\ge 1} \frac{1}{t^3} \mathrm{d}t
\le \eps 
\end{equation}
for $\mu$ larger than $c/\eps^3,$ i.e.,
$\delta_1\ge c/\eps^3\Delta.$
\end{proof}
\begin{lemma}
\label{lem:truncation} 
For every $0<r<k,$
\begin{equation}\label{eq:exptrunc}
\frac1{|\cP|}\E\sum_{P\in\cP_r\backslash\cP_r\trc}|P|
\le \eps\mper
\end{equation}
\end{lemma}
\begin{proof}
We'd like to bound
\begin{equation}\label{eq:trunc}
\frac1{|\cP|}\sum_{P\in\cP_r\backslash\cP_r\trc}|P|
\le
2\sum_{i=1}|\Inf_i(\cP_r)-\E\Inf_i(\cP_r)|
\end{equation}
in expectation over $U_r$ by $O(\eps).$
We will bound $\E|\Inf_i(\cP_r)-\E\Inf_i(\cP_r)|$ for every fixing~$I$ that
includes variable $i$ and fixes all but the variable in position~$r.$
Indeed fix $I=(i_1,\dots,i_{r-1},*,i_{r+1},\dots,i_k).$
Let $\cQ=\{P\in\cP_r\mid \var(P)=I\}.$ We will bound the expected gain of
influence of variable $i$ in $\cQ.$ By linearity of expectation this will give
us a bound on~(\ref{eq:trunc}). Let $Z=\sum_{P\in\cQ\mid \var_r(P)\in
U_r}|P|.$
By Lemma~\ref{lem:degree},
\[
\left|\E Z - Z\right| \le \eps\E Z\mper
\]
Since we have this bound for every fixing and these fixings form a partition of $\cP_r$ we find that after renormalization, we have 
\[
\E|\Inf_i(\cP_r)-\E\Inf_i(\cP_r)|\le\frac\eps n\mper
\]
\end{proof}
Let us denote by $\cP'_r$ the CSP that is obtained in the exact same way as
$\cP_r$ except without the pruning step.  In particular, $\cP'_k$ is simply
the CSP $\cP$ in which the $r$-th variable is restricted to $U_r$ for each
$r\in[k].$

The next corollary summarizes what we have shown so far.
\begin{corollary}
\label{cor:indep-concentration} 
Let $\Psi\subseteq[q]^n$ of size
$\exp(\Omega(\eps^2\delta_1 n)).$ Then,
\begin{equation}
\left|
\Ex{\max_{x\in \Psi}\cP'_k(x)}
-\max_{x\in \Psi}\cP(x) 
\right|
\le
\eps.
\end{equation}
\end{corollary}
\begin{proof}
We first note that $\cP_k\subseteq\cP_k'$ and we can get
\[
\E\frac1{|\cP|}\sum_{P\in\cP_k'\backslash\cP_k}|P|\le \frac\eps2\mper
\]
This follows from repeated application of Lemma~\ref{lem:truncation} (with
sufficiently small value of $\eps$) for each $r\in[k].$ 
In particular this shows that
\begin{equation}\label{eq:exp-max}
\left|\E\max_{x\in[q]^n}\cP_k'(x)
-\E\max_{x\in[q]^n}\cP_k(x)\right|
\le \nfrac\eps2\mper
\end{equation}
On the other hand, by Lemma~\ref{lem:azuma-csp} and the union bound 
over $x\in\Psi$, we get that
\begin{equation}\label{eq:trc}
\left|\Ex{\max_{x\in \Psi}\cP_k(x)}
-\max_{x\in \Psi}\cP(x)\right|\le \nfrac\eps2\mper
\end{equation}
Here we used the fact that the probability that the maximum deviates by
$t\cdot\eps$ drops of exponentially in $t$ so that we can integrate over
$t>1$ to get a bound on the expectation.
Thus,
\[
\left|\Ex{\max_{x\in\Psi}\cP_k'(x)}
-\max_{x\in\Psi}\cP(x)\right| \le\eps\mcom
\]
which is what we wanted to show.
\end{proof}
We are now ready to prove the first main lemma. The proof reduces the general
case to the case where each coordinate is subsampled independently as
previously dealt with. The idea is to partition the set of variables into $m$
bins and only consider predicates whose variables fall into $k$ distinct bins.
The total weight of the remaining predicates can be neglected for large enough
$m.$
\begin{proof}[Proof of Lemma~\ref{lem:randomness} (Concentration)]
Partition $[n]$ randomly into $m$ bins, i.e., 
$[n]=S_1\cup S_2\cup\dots\cup S_m$, with $m=\Theta(1/\eps^2)$ (say).
Furthermore, let $U_\ell=U\cap S_\ell$ for $\ell\in[m].$ One can show that
with probability $1-o(1),$ for all $r\in[k]$ we have $|U_r|\in[\frac1{2m}
|U|,\frac 2m|U|].$

For a given $P\in\cP$ the probability that there are $i,j\in\var(P)$ and
$j\in[m]$ so that $i\in S_\ell$ and $j\in S_\ell$ is at most $O(\eps^2).$
Hence, we can throw away all such $P\in\cP$ and lose only an $O(\eps^2)$
fraction in expectation.

On the other hand, for every $u\in[m]^k$ with pairwise distinct coordinates,
we let $\cP^u=\{P\in\cP\mid \forall r: \var_r\in S_{u_r}\}.$ For every $\cP^u$
we may then apply Corollary~\ref{cor:indep-concentration}, since
$U_{u_1},U_{u_2},\dots,U_{u_r}$ are independently chosen. We apply the
corollary with $\eps'=\eps/m^k=\poly(1/\eps).$ This requires us to choose
$\delta_0$ large enough as a function of $\eps$ so that the previous lemmas
(in particular Lemma~\ref{lem:truncation}) apply to subsets of 
size $|U_r|.$ This allows us to sum the error
over all applications of the Corollary for a total error of~$\eps.$ 
The Corollary applies to sets $\Psi$ of size
$\exp(\Omega(\eps^2|U_r|))=\exp(\poly(\eps)|U|)$ which is what we needed.
\end{proof}

\subsection{Proof of Structure Lemma}

\paragraph{Proof Idea.} 
The main idea is the following. We have a subsample $U$ of size $\delta n.$
Hence, $\max_{x\in[q]^n}\cP_U(x)$ is a maximization problem in $\delta n$
variables. In particular the maximum is achieved by one of 
roughly~$2^{\delta\log(q)n}$ assignments to these variables. 
The whole problem is that we need a set of assignments $\Psi$ of size
$2^{\poly(\eps)\delta n}\ll 2^{\delta n}$ with the property that one of the
assignments in $\Psi$ is near optimal with respect to $\cP_U.$

The proof strategy is to design a deterministic algorithm $D(y)$ that is
given a \emph{seed} $y\in[q]^S$ where $S\subseteq U.$ 
The algorithm returns an assignment $x=D(y)$ to the variables in $U$ with the 
guarantee that for some seed $y\in[q]^S,$ the induced assignment $x=D(y)$ is 
near optimal in $\cP_U.$ An important parameter is the seed length of
$D,$ i.e., the size of $S.$ It is also crucial that the algorithm does not
know~$U$ but only $S$ and $\cP_S.$ (Otherwise the algorithm could trivially
return an optimal assignment for~$\cP_U.$)
Specifically, we want to achieve seed length $|S|\le
\poly(\eps)\delta_1 n/\log(q).$ This will suffice for the purpose of our proof, 
since then we can put $\Psi=\{ D(y) \colon y\in[q]^S\}.$ In this case
$\Psi$ will be sufficiently small.

The key point in the proof is to choose $U$ so large that for every
$x\in[q]^n$ both $\cP_U(x)$ and $\cP_S(x)$ are a good approximation 
of~$\cP(x).$ This fact will be the main reason why we can
hope to obtain a near optimal assignment for $\cP_U$ by just looking at
$\cP_S.$ We remark that this proof strategy is due to~\cite{GoldreichGoRo98}.
Formally we will prove the next lemma.
\begin{lemma}\label{lem:detalg}
For every constant $c,$ there is another constant $C$ and
a deterministic algorithm $D\colon [q]^S\to [q]^U$ which extends an
assignment to the coordinates $S$ to an assignment to the coordinates in $U$
so that
\[
\E\max_{x\in[q]^S}\cP_U(D(x))\ge
\E\max_{x\in[q]^U}\cP_U(x)
-\eps\mper
\]
Here the expectation is taken over random $U\subseteq[n]$ of size 
$\delta n=\eps^{-C}\delta_0 n$ and random subset $S\subseteq U$ of 
size $|S|\le \eps^{-c}\delta n/\log(q).$
\end{lemma}

Once we have this lemma it will be easy to conclude the Structure Lemma. We
will next describe our algorithm and then prove Lemma~\ref{lem:detalg}.

\paragraph{Deterministic greedy algorithm.}

Let $\alpha=\eps^{c_1}/\log(q),$ the factor by which $S$ needs to be smaller
than $U.$ Assume a fixed partition of $U$ into $m$ pairwise disjoint sets
$U=U_1\cup U_2\cup\dots\cup U_m$ of equal size. Here $m$ is some parameter
that we'll need and determine later. Choose $S_\ell$ uniformly at
random from $U\backslash U_\ell$ of size $|S_\ell|=\frac{\alpha}m|U|$ for some
parameter $\alpha.$ Let $S=S_1\cup\dots\cup S_m.$  Note that $|S|=\alpha|U|.$
We want $|S_\ell|\ge\delta_0 n$ so that the concentration lemma will apply
even to the sets $S_\ell.$ We take $m=\poly(1/\eps),$ e.g., $m=\eps^{-2k}$
will be sufficient.  Hence, $\frac\alpha m$ is some fixed polynomial in 
$\eps$ and this determines the size of $U.$

The algorithm $D$ works as follows.

\begin{center}
\begin{boxedminipage}{.98\textwidth}
\noindent Input: $x\in[q]^S$

\noindent Output: $z\in[q]^U$ \vspace{-5mm}
\paragraph{Algorithm:}
\begin{itemize}
\item For every $i\in S,$ we put $z_i=x_i.$
\item For every $\ell\in[m],$ do the following: 
Let $y^*\in[q]^{U_\ell\backslash S}$ denote
the partial assignment that maximizes the function $f(y)=\cP_{S_\ell}(x[y])$ 
where $x[y]$ denote the assignment which is equal to $y$ for all coordinates
$i\in U_\ell\backslash S$ and equal to $x$ elsewhere. Formally, let
\begin{equation}\label{eq:induce}
y^* = \arg\max_{y\in[q]^{U_\ell\backslash S}}
\cP_{S_\ell}(x[y])
\end{equation}
and put $z_i=y^*_i$ for every $i\in U_\ell\backslash S.$
\end{itemize}
This defines an assignment $z$ to all coordinates in~$U.$
\end{boxedminipage}
\end{center}

\paragraph{Analysis.}
\newcommand\err{\mathrm{err}}
\newcommand\errt{\ensuremath{\widetilde{\mathrm{err}}}}
\Tnote{Changed $l$ to $\ell$ in the following (that should be a typesetting issue, I hope I didn't miss any)}
To analyze our algorithm, let $x^*\in[q]^n$ denote the assignment that maximizes $\max_{x\in[q]^n}\cP_U(x).$
Our goal is to define a sequence of ``hybrid'' assignments 
$x^0,x^1,\dots,x^m$
where $x^0=x^*$ and $x^m=D(y)$ for some $y\in[q]^S$
so that in expectation over $U$ and $S,$ we have
\[
\cP_U(x^m)\ge \cP_U(x^0)- \eps\mper
\]
The sequence is defined as follows. Let $x^0=x^*.$ Inductively, let $x^\ell$
for $0<\ell\le m$ be equal to $x^{\ell-1}$ in all coordinates except~$U_\ell$.
The coordinates $U_\ell$ are induced from $S_\ell$ exactly as in our algorithm
in equation~(\ref{eq:induce}), i.e., for all $i\in U_\ell\backslash S,$ we let
$x_i^\ell=y^*_i$ where $y^*$ is defined as
\[
y^* = \arg\max_{y\in[q]^{U_\ell\backslash S}}
\cP_{S_\ell}(x^{\ell-1}[y])\mper
\]
We observe that indeed $x^m$ is the generated by $D$ for some $x\in[q]^S$ since
all coordinates in $x^m$ are induced by looking only at coordinates in~$S$
(though it need not be the case that $x^m=D(x^*)$).
Now, denote the error at step $\ell$ by
\[
\err(\ell)=\cP_U(x^{\ell-1})-\cP_U(x^\ell)\mper
\]
Note that $\err(\ell)$ is a random variable depending on both $U$ and $S.$
The claim now reduces to showing
\begin{equation}\label{eq:total-err}
\E\sum_{\ell\in[m]}\err(\ell)\le \eps\mcom
\end{equation}
since by definition $\cP_U(x^*)-\cP_U(x^m)\le\sum_{\ell\in[m]}\err(\ell)\mper$

In order to argue~(\ref{eq:total-err}), it will be convenient to consider
\[
\errt(\ell)
= \cP_{U\backslash U_\ell}(x^{\ell-1})
- \cP_{U\backslash U_\ell}(x^\ell)\mper
\]
Since we chose $m$ large enough and hence $U_\ell$ is a sufficiently small
fraction of $U,$ it follows that for all $\ell\in[m],$
\begin{equation}\label{eq:total-errt}
\sum_{\ell\in[m]}\err(\ell)
\le 
\sum_{\ell\in[m]}\errt(\ell)
+\frac\eps2\mper
\end{equation}
Let
\[
z^* = \arg\max_{y\in[q]^{U_\ell\backslash S}}
\cP_{U\backslash U_\ell }(x^{\ell-1}[y])\mper
\]
Note that here we are maximizing over $\cP_{U\backslash U_\ell}$ rather than
$\cP_{S_\ell}.$ The following lemma gives us a concrete way of 
bounding $\errt(\ell).$
\begin{lemma}\label{lem:errt-bound}
\begin{equation}
\label{eq:errt-bound}
\errt(\ell)\le
\left|
\cP_{U\backslash U_\ell}(x^{\ell-1}[z^*])
-\cP_{S_\ell}(x^{\ell-1}[z^*]) 
\right|
\end{equation}
\end{lemma}
\begin{proof}
If~(\ref{eq:errt-bound}) were false, then we would have
\[
\cP_{S_\ell}(x^{\ell-1}[z^*]) 
> \cP_{S_\ell}(x^{\ell-1}[y^*]) \mper
\]
But this is a contradiction, since we chose~$y^*$ as the maximum with 
respect to $\cP_{S_\ell}.$ 
\end{proof}
The next lemma shows that the RHS above is small in expectation. The reason is
that $S_\ell$ is chosen uniformly at random inside $U\backslash U_\ell.$
Let $\tilde x=x^{\ell-1}[z^*].$ We have $\E \cP_{S_\ell}(\tilde x)
=\cP_{U\backslash U_\ell}(\tilde x).$ 
We only need to argue that the average deviation of 
$\cP_{S_\ell}(\tilde x)$ from its mean is small.
This can be argued directly but it also follows from 
Lemma~\ref{lem:randomness} applied to $\cP_{U\backslash
U_\ell}$ and $\Psi=\{\tilde x\}.$ To apply this lemma, we actually need that
$\cP_{U\backslash U_\ell}$ is sufficiently close to being sufficiently dense.
This is true in expectation over~$U.$
\begin{lemma}\label{lem:errt}
\begin{equation}
\E\left| \cP_{U\backslash U_\ell}(\tilde x) -\cP_{S_\ell}(\tilde x) 
\right|\le \frac\eps{2m}
\end{equation}
\end{lemma}
\begin{proof}
As mentioned before we think of $\cP_{S_\ell}$ as a subsample of
$\cP_{U\backslash U_\ell}.$ We would like to apply Lemma~\ref{lem:randomness}
(Concentration) to conclude the claim. However $\cP_{U\backslash U_\ell}$ need
not satisfy the density condition. However, by Lemma~\ref{lem:degree},
$\cP_{U\backslash U_\ell}$ does satisfy, for every fixing $I$ of $k-1$
variables,
\begin{equation}\label{eq:approx-deg}
\E_U\left|
\sum_{P\in\cP_{U\backslash U_\ell},\var(P)=I}|P|
-\delta\Delta
\right|\le \eps'\delta\Delta\mper
\end{equation}
In other words, every fixing $I$ satisfies the density requirement in
expectation. We can therefore treat $\cP_{U\backslash U_\ell}$ as a
$\delta\Delta$-dense CSP and subsample $S_\ell\subseteq U\backslash U_\ell$
from it. Note that we can take $\delta\Delta=\poly(1/\eps)$ arbitrarily large
so that we may subsample an $\alpha$ fraction of the variables of
$U\backslash U_\ell$ and expect error $\eps/4m$ in the application of
Lemma~\ref{lem:randomness}. 
The fact that $\cP_{U\backslash U_\ell}$ satisfies
only~(\ref{eq:approx-deg}) leads to additional approximation errors in the
application of Lemma~\ref{lem:randomness}. By summing~(\ref{eq:approx-deg}) 
over all possible~$I,$ we can bound these errors by $\eps'.$ Again taking
$\delta\Delta$ large enough we can assure $\eps'\le\eps/4m.$
Hence, we get a total expected error of $\eps/2m$ which is what we wanted to
show.
\end{proof}
Combining Lemma~\ref{lem:errt-bound} with Lemma~\ref{lem:errt}, we conclude
that for every $\ell\in[m],$
\begin{equation}\label{eq:errt}
\E\errt(\ell)\le\frac\eps{2m}\mper
\end{equation}
Finally,
\begin{align*}
\E\sum_{\ell\in[m]}\err(\ell)
&\le \E\sum_{\ell\in[m]}\errt(\ell)+\frac\eps2
\tag{by~(\ref{eq:total-errt})}\\
& =\sum_{\ell\in[m]}\E\errt(\ell)+\frac\eps2\\
& \le m\cdot \frac\eps{2m} +\frac\eps2 \tag{using~(\ref{eq:errt})}\\
& =   \eps 
\end{align*}
%
We can now complete the proof of the Structure Lemma.
We would like to put 
$\Psi(S)=\{D(x)\mid x\in[q]^S\}.$ Then, by Lemma~\ref{lem:detalg},
\begin{equation}\label{eq:concl}
\E\max_{x\in\Psi(S)}\cP_U(x)\ge
\E\max_{x\in[q]^U}\cP_U(x)
-\eps\mper
\end{equation}
We are not quite done, since 
the Structure Lemma requires a single fixed set $\Psi(S).$ So far
we are choosing $S$ randomly as a subset of~$U.$ Hence, the set $\Psi(S)$ that
we constructed above depends on the choice of~$U.$ To finish the proof we need
a single set $\Psi\sse[q]^n$ that is independent of the choice of~$U.$ This is
easy to accomplish from what we have.  Simply pick $S$ and $U$ independently
and consider $U'=U\cup S.$ Since $|S|\le\poly(\eps)|U|,$ we can make the 
difference between $\cP_U(x)$ and $\cP_{U'}(x)$ negligible for any 
$x\in[q]^n.$ Therefore, we may exchange $U'$ for $U$ in the previous argument 
so that the choice of $S$ and $U$ is independent. 
Since~(\ref{eq:concl}) is then true in expectation taken over independent~$S$
and $U,$ there must also exist a fixed choice of~$S$ for
which~(\ref{eq:concl}) is true in
expectation taken over~$U.$ But now we may take $\Psi=\Psi(S)$ in order to
conclude the proof of the Structure Lemma.

%
%
%
%

\bibliographystyle{alphasy1}
\bibliography{approx}

\appendix

\section{Edge distribution of subsample of the third power}
\label{sec:proxy}
In this section we compare the edge distribution of the subsample of~$G^3$ to
a somewhat nicer distribution. This step was needed in
Lemma~\ref{lem:second-third-step-2}.
In the following let $G=(V,E)$ be a $\Delta$-regular graph and
$\delta\ge\poly(\eps^{-1})\Delta^{-1}.$ Further denote $W=V_\delta.$
\begin{lemma}
Let $D_1$ denote the uniform distribution over edges in $G^3[W].$
Let $D_2$ denote the distribution obtained as follows:
\begin{enumerate}
\item Pick a random edge $(v,v')\in E.$
\item Choose uniformly at random $w\in N_W(v)$ and $w'\in N_W(v').$
\item Output $(w,w').$
\end{enumerate}
Then,
\[
\E_W\left[{\rm TV}(D_1,D_2)\right]\le\eps\mper
\]
Here and in the following
${\rm TV}(D_1,D_2)$ denote the total variation distance between the
two distributions~$D_1$ and~$D_2.$
\end{lemma}
\begin{proof}
Let us compare the following two distributions:
\begin{description}
\item[$P_1:$] Pick a uniformly random path $p=(w,v,v',w')$ from the set of
all paths of length~$3$ in $G$ which have $w,w'\in W.$
\item[$P_2:$] Pick a random edge $v,v'\in E$ and random neighbors $w\in
N_W(v),w'\in N_W(v')$ and consider the path $(w,v,v',w').$
\end{description}
Notice that it suffices to bound the statistical distance between $P_1$ and
$P_2.$ This is because $D_1$ is just the marginal distribution of $P_1$
on the endpoints of the path~$(w,w').$ Likewise $D_2$ is the marginal
distribution of $P_2$ on~$(w,w').$

Now, let $p=(w,v,v',w')$ denote any path of length~$3$ in $G$ so that
$w,w'\in W.$ Let $N$ denote the number of such paths. Note that $\E N = \delta^2\Delta^3 n.$
Let us now
compare the probability of this path under the two distributions.
For $P_1$ we get
\[
P_1(p)
= \frac 1N\mper
\]
On the other hand, under $P_2,$
\[
P_2(p) = \frac1{|N_W(v)|}\cdot \frac1{\Delta n} \cdot \frac1{|N_W(v')|}.
\]
Note that for every $v\in V,$ we have $\E|N_W(v)|=\delta\Delta.$
It now suffices to argue the bound
\begin{equation}
\Ex{{\rm TV}(P_1,P_2)}
=
\E\frac12\sum_p \left|
\frac1{N}-
\frac1{|N_W(v)||N_W(v')|\Delta n}\right|\le \eps.
\end{equation}
Let us call a path~$p=(w,v,v',w')$ \emph{good} if
\[
\frac1{|N_W(v)|\cdot |N_W(v')|}
=\frac{1\pm\eps'}{\delta^2\Delta^2}.
\]
Later we will choose $\eps'=\Omega(\eps)$ to be sufficiently small, say,
$\eps'=\eps/100$. We need
the following simple concentration bounds.
\begin{claim}
With probability $1-\eps'$ over the choice of $W,$ we have
\begin{enumerate}
\item
$N^{-1} = (1\pm\eps')/(\delta^2\Delta^3 n)\mper$
\item
The fraction of bad paths is less than
$1/O(\eps'^5(\delta\Delta)^3)\mper$
\end{enumerate}
\end{claim}
\begin{proof}
The first claim follows from Lemma~\ref{lem:edgeweight}. Regarding the second
claim, it is not hard to show for every $v,v'$ that
\[
\Pr\left\{\frac1{|N_W(v)||N_W(v')|}\not\in\frac{1\pm\eps'}{\delta^2\Delta^2}\right\}
\le \frac1{O(\eps'^4(\delta\Delta)^3)}\mper
\]
This can be shown by computing the fourth moment $\E(|N_W(v)|-\delta\Delta)^4$
and bounding the probability of a factor $1+\alpha$ deviation of $|N_W(v)|$
from its mean for small enough~$\alpha=\Omega(\eps').$
This argument shows that the expected number of bad paths is at most
$1/O(\eps'^4(\delta\Delta)^3)$ and the claim is completed by applying
Markov's inequality.
\end{proof}
Given this claim, we can finish the proof of the lemma. Indeed letting $Q$
denote the set of good paths, we have with
probability $1-\eps',$
\begin{align*}
\sum_p \left|
\frac1{N}-
\frac1{|N_W(v)||N_W(v')|\Delta n}\right|
& \le \sum_{p\in Q} \frac{2\eps'}{\delta^2\Delta^3n}
+ \sum_{p\not\in Q} \frac{1}{\Delta n}\\
& \le 2\eps' + \frac N{O(\eps'^5(\delta\Delta)^3)}\cdot\frac1{\Delta n}\\
& = 2\eps' + \frac1{O(\eps'^5\delta\Delta)}\cdot\frac N{\delta^2\Delta^3 n}\\
& \le O(\eps')\mper
\end{align*}
In the first inequality we used the fact that $|N_W(v)|\ge1$ for any existing path
and hence the term $1/|N_W(v)||N_W(v')|\Delta n$ is never larger than~$1/\Delta
n.$
In the last step we used that we may choose $\delta\Delta\ge C\eps'^{-5}$ for
sufficiently large constant $C>0,$ and that $N\le(1+\eps')\delta^2\Delta^3n.$
Hence,
\[
\E{\rm TV}(P_1,P_2)\le (1-\eps')O(\eps') + \eps'\le\eps\mper
\qedhere
\]
\end{proof}

\section{Details on random geometric graphs~\ref{sec:maxcut}}
\label{sec:maxcut-details}

In this section we will in the details that were left out in
Section~\ref{sec:maxcut}. We start with the proof of Lemma~\ref{lem:sdp3}.

\restate{lem:sdp3}

The proof works as follows. First, triangle inequalities are
known to imply the odd cycle constraints which means that an SDP with triangle
inequalities on an odd cycle of length $k$ has value at most (and, in fact,
equal to)~$1-1/k.$

\begin{lemma}\label{lem:odd}
Let $C$ be an odd cycle of length~$k.$ Then, $\sdpGWT(C)\le 1-1/k.$
\end{lemma}

Second, it follows that if a graph $G$ can be covered uniformly by odd cycles
of length $k,$ then its $\sdp_3$-value can be at most $1-1/k.$

\begin{lemma}
Let $G=(V,E)$ be a (possibly infinite) graph.
Suppose there exists a distribution~${\cal C}$ over odd
cycles of length~$k$ for some fixed number $k$ such that the marginal
distribution on each edge of a random cycle from ${\cal C}$ has statistical
distance $\eps$ to the uniform distribution over edges in $G.$  Then,
 $\sdpGWT(G)\le 1-1/k+\eps.$
\end{lemma}
\begin{proof}
By our assumption we have that for every embedding $f\colon V\to B,$
\[
\E_{(u,v)\sim E}\tfrac14\|f(u)-f(v)\|^2
\le
\E_{C\sim{\cal C}}
\E_{(u,v)\sim C}\tfrac14\|f(u)-f(v)\|^2
+\eps\mper
\]
But we know, by Lemma~\ref{lem:odd}, that for every $f\colon V\to B,$
satisfying the triangle inequalities,
\[
\E_{(u,v)\sim C}\frac14\|f(u)-f(v)\|^2\le 1-\tfrac1k\mper
\]
Hence,
\[
\E_{(u,v)\sim E}\frac14\|f(u)-f(v)\|^2\le 1-\tfrac1k+\eps\mper
\]
\end{proof}
We will next see that the sphere graph can by uniformly covered by odd cycles
of length $O(1/\sqrt{\gamma}).$ We begin with the following simple observation.
\begin{lemma}\label{lem:cycle}
For every $l\in[1-\gamma,1-\gamma/2],$ there exists an odd cycle,
denoted $C_l=(v_1,\dots,v_k)$, in $G_\gamma$ of length~$k=O(\sqrt{\gamma})$
such that $\frac14\|v_i-v_{i+1}\|^2=l$ for all $i\in\{1,\dots,k-1\}.$
\end{lemma}
\begin{proof}[Proof sketch.]
Pick an arbitrary great circle around the sphere and place the
vertices $v_1,\dots,v_k$ equally spaced along this circle.
For $k=O(\sqrt{\gamma})$ vertices, we can accomplish the Euclidean distance
between two consecutive vertices is less than, say, $\sqrt{\gamma}/10.$ Now
connect each vertex $v$ on the circle to the unique vertex $w$ which maximized
$\|v-w\|^2.$ This creates an odd cycle and, by our previous observation,
it follows that $\frac14\|v-w\|^2\ge1-\gamma.$ Now we can make
$\frac14\|v-w\|^2=l$ be walking along the cycle and moving vertices in a
direction orthogonal to the plane defined by the circle until all edges have
length~$l.$
\end{proof}
\begin{lemma}\label{lem:continuousDual}
  Let $\gamma>0$ and let~$S^{d-1}$ be the sphere.  There exists a
  distribution~${\cal C}$ over odd cycles~$C=(v_1,\ldots,v_k)$ for some
  $k \le \frac{10\pi}{\sqrt{\gamma}}$ such that for all $i$,
  the marginal distribution of $(v_i, v_{i+1})$ has statistical
  distance $o(1)$ to the uniform distribution over edges in $G_\gamma$
(as~$d\to\infty$).
\end{lemma}
\begin{proof}
We will describe the distribution ${\cal C}$ as follows:
\begin{enumerate}
\item
Pick a random edge $e=(u,v)\in E$ from~$G_\gamma.$
\item\label{step2}
Let $l=\frac14\|u-v\|^2.$ If $l\le1-\gamma/2,$ let $C_l=(v_1,v_2,\dots,v_k)$
denote the odd cycle given by Lemma~\ref{lem:cycle}.
If $l\ge 1-\gamma/2,$ declare ``failure''.
\item
If the previous step succeeded,
pick a random rotation~$R$ and output $RC=(Rv_1,Rv_2,\dots,Rv_k).$
\end{enumerate}

We claim that if the second step succeeds, then indeed every marginal
$(Rv_i,Rv_{i+1})$ is distributed like a uniformly random edge. This is (1)
because $(u,v)$ was chosen to be a uniformly random edge and (2)
$(Rv_i,Rv_{i+1})$ is a random rotation of~$(u,v)$ and hence, by spherical
symmetry, is equally likely to
be any edge in $E$ that has the same length as $(u,v).$

On the other hand, by measure concentration, with probability
$1-\exp(-\Omega(d)),$ we have that $\frac14\|u-v\|^2\in[1-\gamma,1-\gamma/2].$
This completes the claim since the probability of failure only introduces
$o(1)$ statistical distance.
\end{proof}

In this section we give some details on how to obtain a dense discretization
of the Feige-Schechtman graph.

\restate{lem:dense}

\begin{proof}[Proof sketch.]
The first claim is shown in \cite{FeigeSc02}. For the second claim,
let us decompose $\mathbb{S}^{d-1}$ into equal volume cells of diameter at
most~$\epsilon.$ Here, $\epsilon$ is a parameter that we will later take to be
very small, say, $\epsilon\le d/100.$
Now pick enough vectors $V\subseteq\mathbb{S}^{d-1}$ uniformly at
random such that with probability at least $1-\epsilon$ every two cells have
the same number of vectors up to a factor of $1\pm\epsilon$ in it.

We need to show that $\sdp_3(G_\gamma[V])\le1-\Omega(\sqrt{\gamma}).$
To this end we first consider a related graph $G'$, which has the same vertex
set as $G$ but different edges. A random edge in $G'$ is defined by the
following process: Pick first a random edge on the continuous sphere, then for
each endpoint pick a random vertex in the equal volume cell containing the
endpoint. Finally, normalize the edges such that the total edge weight is the
same as in $G$.

We can use the distribution over odd cycles given by
Lemma~\ref{lem:continuousDual} in order to get a distribution for the graph
$G'$ as follows: Pick the cycle and map each point to a vertex in the
corresponding cell. The resulting marginal distributions will be uniform in
$G'.$ Thus, $\sdp_3(G')=1-\Theta(\sqrt{\gamma}).$

Finally, we will show that $\E\normtv{L(G')-L(G_\gamma[V])}$ tends to zero
with~$\eps.$ That is, the two distributions have statistical distance tending
to zero. This also shows that for sufficiently small~$\eps$, the
semidefinite programs also have approximately the same value.
Now to argue the above point, consider the process of picking a random edge.
Consider first the case that in $G'$, the two cells containing the chosen points
have exactly the expected number of vectors in them, and furthermore, suppose
that the two cells are \emph{good} in the sense that either none of the
vertices in them share an edge or all pairs of vertices between the two
cells share an edge in~$G.$
In this case, the edges in $G$ going between these two cells have exactly the same
probability as under $G'$.

The first assumption is close enough to the truth,
since the number of vertices in different cells differ by at most a factor of $1\pm \epsilon$,
For the second assumption it suffices to pick~$\eps$ small enough so that a
cap of radius $r$  has the same volume as a cap of radius $r\pm\eps$ up to
a factor of $1\pm o(1).$ This happens for, say, $\eps\ll1/d.$ This will
guarantee that the number of bad pairs of cells is small. This argument can be
found in~\cite{FeigeSc02}.
\end{proof}

%
%
%

\section{Subsampling edges}
In this section, we will briefly discuss the analogue of our main theorem in
the setting where we sample a fraction of the edges in~$G$ at random so that
the expected degree in $G$ is constant.
Here, $G=(V,E)$ will always denote a $\Delta$-regular
graph on $n$ vertices.
Our proof in the case of edge subsampling is much simpler. As it turns out it
suffices to bound the cut norm between the original graph and its subsample
and to argue that the SDP value is a Lipschitz function of the cut norm. The
latter fact is a consequence of Grothendieck's inequality.

We let $E_\delta\subseteq E$ denote a random subset of $E$ of size
$\delta|E|.$ We'll overload notation slightly by using $G[E_\delta]$ for the
graph $G$ restricted to the edge set $E_\delta.$
\begin{definition}
The cut norm of a real valued $n\times n$ matrix $A$ is defined as
\begin{equation}\label{eq:cutnorm}
\|A\|_C
=\max_{U,V\subseteq[n]}\left| \sum_{i\in U,j\in V} a_{ij} \right|.
\end{equation}
\end{definition}
It is known that the cut norm is within constant factors of the norm
\begin{equation}\label{eq:infty}
\|A\|_{\infty\mapsto 1} = \max_{x_i,y_j\in\{-1,1\}} \sum_{i,j\in[n]}
a_{ij}x_iy_j\mper
\end{equation}
A natural semidefinite relaxation of~(\ref{eq:infty}) replaces every pair $x_i,y_j$
by two unit vectors $u_i,v_j,$ i.e.,
\begin{equation}\label{eq:inftysdp}
\sdp_C(A) = \max_{\|u_i\|=\|v_i\|=1} a_{ij}\langle u_i,v_j\rangle.
\end{equation}
A theorem of Grothendieck bounds the gap between the cut norm and its relaxation by a
multiplicative constant (the Grothendieck constant).
\begin{theorem}
\label{thm:groth}
There is a constant $K_G$ (known to be less than~$1.8$) such that
$\sdp_C(A) \le K_G \|A\|_{\infty\mapsto 1}.$
\end{theorem}
The next lemma shows that the cut norm between a graph and its subsample is small.
\begin{lemma}\label{lem:cutnorm}
Let $\delta\ge c\eps^{-2}\Delta^{-1}.$ Then,
$\E\left\|A(G)-\delta^{-1} A(G[E_\delta])\right\|_{\infty\mapsto 1}\le\eps.$
\end{lemma}

\begin{proof}
We can show that with probability $1-e^{\Omega(n)}$, $|\langle
x,Ay\rangle - \delta^{-1}\langle x,A'y\rangle|\le\eps$
simultaneously for all $x,y\in\{-1,1\}^n.$ The proof follows from Hoeffding's
bound and the union bound. The details are straightforward and
therefore omitted from this paper.
\end{proof}
Similarly the following lemma can be shown.
\begin{lemma}\label{lem:degrees}
Let $\delta\ge c\eps^{-2}\Delta^{-1}.$ Then,
$\E\left\|D(G)-\delta^{-1}D(G[E_\delta])\right\|_{\infty\to1}\le\eps\mper$
\end{lemma}
The previous two lemmas showed that the expected difference in cut norm between the graph $G$ and
its edge subsample $G[E_\delta]$ is small.
\begin{corollary}
For $\delta\ge c\eps^{-2}\Delta^{-1},$ we have
$\E\|L(G)-\delta^{-1}L(G[E_\delta])\|_C\le \eps\mper$
\end{corollary}
It turns out that bounding the difference in cut norm is sufficient for bounding the difference in
SDP values.
\begin{lemma}\label{lem:LipschitzSDP}
Let $G$ and $G'$ be any two graphs on $n$ vertices. Let $\cM\subseteq\cM_2$
(see Definition~\ref{def:reasonable}) be any set of positive semidefinite
$n\times n$ matrices. Suppose $\|L(G)-L(G')\|_C\le t$. Then,
\[
|\sdp_\cM(G)-\sdp_\cM(G')|\le O(t)\mper
\]
\end{lemma}

\begin{proof}
\begin{align*}
\left|\sdp_\cM(G)-\sdp_\cM(G')\right|
&\le\left|\,\max_{X\in\cM_2}(L(G)-L(G'))\bullet X\,\right|\\
&\le O(1)\cdot \|L(G)-L(G')\|_C \tag{by Theorem~\ref{thm:groth}}\\
&\le O(t)\mper
\end{align*}
\end{proof}

\begin{corollary}
Let $G$ denote a $\Delta$-regular graph and let $\delta\ge\poly(1/\eps)\Delta^{-1}.$ Then,
\begin{equation}
\E\left|\sdp_\cM(G)-\sdp_\cM(G[E_\delta])\right|\le\eps
\end{equation}
for any $\cM\subseteq\cM_2.$
\end{corollary}

\paragraph{Negative results for linear programs.}
We remark that using the approach of~\cite{CharikarMaMa09} one can obtain
strong and general results ruling out subsampling for linear programs.
\begin{theorem}\label{thm:maxcut-sa}
Let $\eps,\lambda>0.$ Suppose $G$ is a $\Delta$-regular graph with $\Delta>n^\theta.$ Then with
high probability over $G'=G[E_{\lambda/\Delta}],$ after removing $o(n)$ vertices,
$\sa_r(G')\ge1-\eps$
for $r=n^\alpha$ where $\alpha(1/\eps,1/\theta,\lambda)$ tends to zero as any of its arguments
grows.
\end{theorem}
The proof follows by arguing that $G[E_{\lambda/\Delta}]$ has sufficient small
set expansion so that~\cite{CharikarMaMa09} applies. Details are omitted.

\section{Deviation bounds}
\label{sec:deferred}

\paragraph{Deviation bounds for submatrices.}
The following general lemma is useful in bounding the deviation of expressions
$\sum_{i,j\in S} |a_{ij}|$ when $S$
denotes a random subset of~$[n]$ and $A$ is a $n\times n$ matrix.

\begin{lemma}\label{lem:edgeweight}
Let $A$ denote a symmetric $n\times n$ matrix such that $a_{ii}=0$ for all $i\in[n].$
Suppose there is some $\beta>0$ such that $|a_{ij}|\le\beta$ for all
$i,j\in[n]$ and $\sum_{j}|a_{ij}|\le 1$ for all $i.$
Now, let $S\subseteq[n]$ denote a random subset of $[n]$ of size
$\delta n$ for some $\delta>\beta.$ Then, for all $\eps>0,$
\begin{equation}
\Pr\left(\left|\delta^{-2}\sum_{i,j\in S}
a_{ij}-\sum_{i,j\in[n]}a_{ij}\right|>\eps n\right)
\le \frac{O(1)}{\eps^2\delta n}\mper
\end{equation}
\end{lemma}

\begin{proof}
Denote by $X_{ij}$ the random variable which is equal to $a_{ij}$ when both
$i\in S$ and $j\in S$ and is zero otherwise. Let $\mu_{ij}=\E
X_{ij}=\delta^2a_{ij}.$ Putting $X=\sum_{i,j\in[n]}X_{ij}$ and $\mu=\E X$ we
will compute the variance of $X.$ The key fact that we will use is that the
selection of $i,j$ and $k,l$ is independent unless either $i=k$ or $j=l.$
Pairs where neither is the case will not contribute to the variance. More
precisely,
\begin{align*}
\E\left( X - \mu \right)^2
& = \E \left( \sum_{ij} X_{ij} - \mu_{ij} \right)^2\\
& = \Ex{ \sum_{i,j,k,l} (X_{ij}-\mu_{ij})(X_{kl}-\mu_{kl}) }\\
& = \sum_{ij} \E\left(X_{ij}-\mu_{ij}\right)^2
  + \sum_{ijk} \E\left(X_{ij}-\mu_{ij}\right)\left(X_{kj}-\mu_{kj}\right)\\
& = \sum_{ij} O(\delta^2)a_{ij}^2
  + \sum_{ijk} O(\delta^3)a_{ij}a_{kj}\mper
\end{align*}
At this point notice that $\sum_{ij}a_{ij}^2$ is maximized when in every row
we have $1/\beta$ entries of magnitude $\beta$ in which case the expression
evaluates to $\frac1\beta\beta^2 n=\beta n.$ Likewise the second
expression $\sum_{ijk} a_{ij}a_{kj}$ is maximized when in every column
$j\in[n]$ we have $1/\beta$ nonzero entries of magnitude $\beta.$ In this case
the expression is $(1/\beta)^2\beta^2 n = n.$ Hence,
\begin{equation}
\sigma^2= \E\left( X - \mu \right)^2
\le O(\delta^2 \beta n) + O(\delta^3 n)
\le O(\delta^3 n)
\mcom
\end{equation}
where we used that $\delta>\beta.$
Hence by Chebyshev's inequality,
\[
\Pr(|X-\mu|\ge\eps\delta^2 n)
\le\frac{\sigma^2}{\eps^2\delta^4n^2}=\frac{O(1)}{\eps^2\delta n}\mper
\]
This is what we claimed up to scaling.
\end{proof}

In the proof of the proxy graph theorem we used the following simple
observation relating the Laplacian of a subsample~$L(G[V_\delta])$
to the corresponding principal submatrix of the Laplacian~$L(G)_{V_\delta}.$
\begin{lemma}
\label{lem:minor-vs-sub}
Let $G$ be a $\Delta$-regular graph and let $H$ be a graph of degree at least
$\Delta.$ Let $\delta\ge\poly(1/\eps)\Delta^{-1}.$ Then,
\[
\E\normtv{L(G[V_\delta]) -L(G)_{V_\delta}}
\leq \epsilon \mper
\]
\end{lemma}
\begin{proof}
By inspection of the two matrices we see that the difference in the
entries of the matrix is due to irregularities in the degrees
of $G[V_\delta].$ Specifically, the matrix $L(G)_{V_\delta}$ has diagonal
entries equal to $1/\delta n.$ On the other hand, the $i$-th diagonal entry of
$L(G[V_\delta])$, call it $d_i$, is equal to
$\delta^{-1}\sum_{j\in V_\delta} a_{ij}.$ We have that
$\E d_i = \frac1{\delta n}$ and we claim,
\[
\sum_{i\in V_\delta}\E\left|d_i-\frac1{\delta n}\right|\le\eps\mper
\]
This can be derived from Lemma~\ref{lem:edgeweight}.
\end{proof}


\paragraph{McDiarmid's inequality.}
We also needed McDiarmid's large deviation bound (sometimes called
Azuma's inequality).
\begin{lemma}\label{lem:mcd}
Let $X_1,\dots,X_m$ be independent random variables all taking values in the
set ${\cal X}$. Further, let $f\colon{\cal X}^m\to\mathbb{R}$ be a function of
$X_1,\dots,X_m$ that satisfies for all $i,x_1,x_2,\dots,x_m,x_i'\in{\cal X}$,
\[
|f(x_1,\dots,x_i,\dots,x_m)-f(x_1,\dots,x_i',\dots,x_m)|\le c_i.
\]
Then, for all $t>0$,
\[
\Pr\{ |f-\Ex{f}|\ge t\}\le
2\exp\left(\frac{-2t^2}{\sum_{i=1}^mc_i^2}\right).
\]
\end{lemma}
%
%

\end{document}